\documentclass[11pt]{article}
\usepackage[T1]{fontenc}
\usepackage[utf8]{inputenc}
\usepackage{lmodern}
\usepackage{textcomp}
\usepackage{mleftright}
\usepackage[margin=1in]{geometry}
\usepackage{mathtools, amsmath, amsthm, thm-restate, graphicx, enumitem, amsfonts, amssymb, verbatim, bbm, braket, xcolor, float, complexity, comment, tikz, tikz-3dplot, url, mdframed, needspace, titlesec}
\usepackage[linktocpage=true]{hyperref}
\usepackage[backend=biber, style=alphabetic, maxbibnames=99]{biblatex}
\usepackage[linesnumbered, ruled, vlined]{algorithm2e}
\usepackage[capitalise]{cleveref}
\usepackage[most]{tcolorbox}
\allowdisplaybreaks

\usepackage{makecell}
\newmdenv[linecolor=black, linewidth=1pt]{factbox}

\newcommand{\lambdaMax}[1]{\lambda_{\mathrm{max}}\mleft(#1\mright)}
\newcommand{\knote}[1]{{\color{brown} (Kunal: #1)}}
\newcommand{\jnote}[1]{{\color{teal} (James: #1)}}
\newcommand{\anote}[1]{{\color{orange} (Anuj: #1)}} 
\newcommand{\enote}[1]{{\color{blue} (Eunou: #1)}}  
\newcommand{\ratioUB}{0.839512}
\newcommand{\onote}[1]{{\color{magenta} (Ojas: #1)}}
\newcommand{\lnote}[1]{{\color{purple} (Lennart: #1)}}
\let\oldnl\nl
\newcommand{\nonl}{\renewcommand{\nl}{\let\nl\oldnl}}

\newcommand{\td}[1]{\widetilde{#1}}
\newcommand{\defeq}{\stackrel{\mathrm{\scriptscriptstyle def}}{=}}
\newcommand{\of}[1]{\left( #1 \right)}
\newcommand{\setFunct}[2]{\left\{ #1 \, \middle| \, #2 \right\}}
\newcommand{\ofc}[1]{\left\{ #1 \right\}}
\newcommand{\ofb}[1]{{\left[#1\right]}}
\newcommand{\imagUnit}{\mathbf{i}}

\newcommand{\StoqMA}{\class{StoqMA}}

\newtheorem{definition}{Definition}
\newtheorem{claim}{Claim}
\newtheorem{lemma}{Lemma}
\newtheorem{theorem}{Theorem}
\newtheorem{corollary}{Corollary}

\hypersetup{
    colorlinks=true,
    linkcolor=blue,
    filecolor=magenta,      
    urlcolor=cyan,
    citecolor=violet,
}

\begin{document}
\title{A 0.8395-approximation algorithm for the EPR problem}
\author{ 
  Anuj Apte\thanks{Global Technology Applied Research, JPMorganChase} \quad
  Eunou Lee\thanks{Korea Institute for Advanced Study} \quad
  Kunal Marwaha\thanks{University of Chicago}\footnotemark[3] \\
  Ojas Parekh\thanks{Sandia National Laboratories} \quad
  Lennart Sinjorgo\thanks{CentER, Department of Econometrics and OR, Tilburg University} 
    \quad
  James Sud\footnotemark[3]
  \thanks{\texttt{jsud@uchicago.edu}}
}

\date{\today}
\maketitle

\begin{abstract}
    We give an efficient 0.8395-approximation algorithm for the EPR Hamiltonian. Our improvement comes from a new nonlinear monogamy-of-entanglement bound on star graphs and a refined parameterization of a shallow quantum circuit from previous works. We also prove limitations showing that current methods cannot achieve substantially better approximation ratios, indicating that further progress will require fundamentally new techniques.
\end{abstract}

\section{Introduction}\label{sec:intro}

Consider a graph $G=(V,E,w)$ on vertices $V \defeq \ofb{n}$, with edges $E\subseteq V\times V$ and edge weights $w \in \mathbb{R}_{> 0}^E$. The \emph{EPR problem}, as introduced by \cite{king2023}, is to find the maximum eigenvalue of the following Hamiltonian:
\begin{equation}
\label{eqn_eprHamiltDef}    
\begin{aligned}
     H(G) &\defeq \sum_{(i,j) \in E} w_{ij} h_{ij}\,, \\
     h_{ij} &\defeq \frac{1}{2} (I_iI_j +X_iX_j - Y_iY_j +Z_iZ_j)\,.
\end{aligned}
\end{equation}
Note that the off-diagonal elements of $h_{ij}$ are all positive in the computational basis, so the decision version of the EPR problem on a graph with positive weights is in $\StoqMA$ \cite{bravyi2006}. We do not know whether it is in $\P$. 
We study $\alpha$-approximation algorithms for the EPR problem. Such an algorithm $\mathcal{A}$ inputs a graph $G$, and outputs a value $\mathcal{A}(G)$ such that $\alpha \cdot \lambda_{\max}(H(G)) \le \mathcal{A}(G) \le \lambda_{\max}(H(G))$. An efficient $\alpha$-approximation algorithm was first shown in \cite{king2023} for $\alpha > 0.7071$. This was improved to $\alpha = 0.72$ in \cite{jorquera2024}, and later to $\alpha > 0.8090$ in \cite{apte2025, ju2025}.  Our main result is a further improvement of $\alpha$ to $> 0.8395$:
\begin{theorem}
\label{thm:apx_ratio}
    There is an efficient $\alpha$-approximation algorithm for the EPR problem for $\alpha > 0.8395$.
\end{theorem}

\subsection{Algorithm}
Our algorithm uses the same quantum circuit structure (ansatz) as previous approximation algorithms for the EPR problem~\cite{anshu2020,king2023,apte2025}. This ansatz does the following:
\begin{enumerate}
    \item Solve an efficient relaxation of the EPR problem to obtain values $g_{ij}$ for each edge $(i,j)$ in $G$. 
    \item Apply a depth-$1$ quantum circuit parameterized by a set of angles $\{\theta_{ij}\}_{(i,j) \in E}$. The angles are chosen as a deterministic function of $\{g_{ij}\}_{(i,j) \in E}$.
\end{enumerate}
For convenience, we assume that $\theta_{ij}$ depends \emph{only} on $g_{ij}$; i.e. $\theta_{ij} = \nu(g_{ij})$ for some function $\nu$. We choose values $\ofc{g_{ij}}_{(i,j)\in E}$ by solving the quantum moment sum-of-squares (moment-SoS) hierarchy defined in \cref{sec:sos/def}. We state this ansatz formally as \Cref{alg:epr}:

\begin{algorithm}[ht]
\caption{EPR approximation ansatz}\label{alg:epr}
    \vspace{1ex}
    \nonl \textbf{Input:} weighted graph  $G(V, E, w)$ and function $\nu: \ofb{-1,1} \rightarrow \ofb{0,1}\,$\vspace{1ex}\\
    Solve the level-$2$ quantum moment-SoS hierarchy (\cref{eq:epr_sdp}) to obtain $\ofc{g_{ij}}_{(i,j)\in E}\,$. \\
    Output the state
    \begin{align}\label{eq:epr_chi}
        \ket{\psi_G} \defeq \prod_{(i,j) \in E} \exp\of{\frac{\imagUnit \,\theta_{ij}}{4} (X_i - Y_i)\otimes(X_j - Y_j)} \ket{0}^{\otimes n}\,,
    \end{align}
    where $\imagUnit$ is the imaginary unit and $\theta_{ij} = \nu(g_{ij})$.
\vspace{.5em}
\end{algorithm}

\Cref{alg:epr} comes with a classically computable lower bound on its average energy:
\begin{lemma}[{\cite[Lemma 9]{king2023}}]
\label{lem:king_simplification}
\Cref{alg:epr} prepares a state with energy at least
\begin{align}\label{eq:single_edge_energy}
    \bra{\psi_G} H(G)\ket{\psi_G} \ge \ell(G) \defeq \sum_{(i,j) \in E} w_{ij} \cdot \frac{1 + A_{ij} A_{ji} + (A_{ij} + A_{ji}) \sin \nu(g_{ij})}{2}\,,
\end{align}
where 
$
A_{ij} \defeq \prod_{k \in N(i)\setminus\{j\}} \cos \nu(g_{ik})\,
$ and $N(i)$ is the set of neighbors of $i$ in $G$. Furthermore, the inequality is an equality for all triangle-free graphs.
\end{lemma}
To fully specify \Cref{alg:epr}, we must choose a function $\nu$. We provide our specific choice of $\nu$ in \cref{thm:apx_ratio_formal}. This function was obtained via numerical search, and we prove the correctness of the approximation ratio analytically in \cref{thm:apx_ratio_formal}. We show in \Cref{sec:limits} that our choice of $\nu$ is essentially optimal.

\subsection{Techniques}

Our improved approximation ratio relies on two crucial insights. First, we prove a new \emph{monogamy of entanglement} (MoE) statement, generalizing the bound of \cite[Lemma 3]{lee2024} from a pair of edges to a star. As with other MoE bounds, we show this statement holds for the EPR problem \textit{and} for its level-$k$ semidefinite relaxation whenever $k \ge 2$. We describe these relaxations in \Cref{sec:sos} and prove the lemma in \cref{apx:omitted_proofs/degree_moe}. 

\begin{lemma}[Nonlinear monogamy of entanglement on a star]\label{lem:degree_moe}
Fix any graph $G=(V,E,w)$ and state $\ket{\psi_G}$. For each edge $(i,j) \in E$, let $g_{ij} = \bra{\psi_G} h_{ij} \ket{\psi_G} - 1$. Then for any vertex $i\in V$ with degree $d_i \geq 2$ and $j \in N(i)$,
\begin{equation}
\label{eq:degree_moe}
\begin{aligned}
    \sum_{k \in N(i) \setminus \ofc{j}} g_{ik} &\le \begin{cases}
        1, \quad &\text{ if } \, -1 \le g_{ij} < -\frac{1}{d_i}, \\
        \frac{1}{2}\of{2-d_i-g_{ij} + \sqrt{\of{d_i^2-1}\of{1-g_{ij}^2}}}, \quad &\text{ if } \,-\frac{1}{d_i} \leq g_{ij} \le 1 \,,
    \end{cases}
\end{aligned}
\end{equation}
where $N(i)$ is the set of neighbors of $i$. Moreover, \cref{eq:degree_moe} holds for any $g$ that is the solution to the level-$k$ semidefinite relaxation of $H(G)$ (defined in \Cref{sec:sos}), for any $k \ge 2$.
\end{lemma}

Our second insight follows from \cite{gribling2025}. We choose the parameterization function $\nu$ from a set $\mathcal{C}$ that dramatically simplifies the approximation ratio analysis:
\begin{definition}\label{def:theta_set}
    Let $\mathcal{C}$ be the set of monotonically increasing functions $\Theta: [0,1] \to [0,1]$ where $\Theta(0) = 0$, and for all $x_1,x_2,\ldots,x_p \in \ofb{0,1}$ satisfying $\sum_{i=1}^p x_i \le 1$, we have
        \begin{align}
        \label{eqn:theta_condition}
            \prod_{i=1}^p \of{1-\Theta(x_i)} \ge 1- \Theta\of{\sum_{i=1}^p x_i}.
        \end{align}
\end{definition}
We sketch how the set $\mathcal{C}$ simplifies our analysis. Suppose $\nu(x) = \arcsin \sqrt{ \Theta(x^+)}$ for some $\Theta \in \mathcal{C}$, where we use the notation $x^+\defeq \max\{x,0\}$. Then, the energy from \cref{lem:king_simplification} becomes
\begin{align*}
A_{ij} = \prod_{k \in N(i) \setminus \{j\}} \sqrt{1 - \Theta(g_{ik}^+)} \ge \sqrt{1 - \Theta\of{\sum_{k \in N(i) \setminus \{j\}} g_{ik}^+}}\,.
\end{align*}
Through our parameterization of $\nu$, we have converted a product of trigonometric functions to a sum over $g^+$ values on neighboring edges. We may then directly apply MoE bounds on these $g$ values to lower bound $A_{ij}$. 

Finally, in \cref{sec:limits}, we discuss natural limitations of our ansatz and analysis. For example, one step of our analysis considers the approximation ratio on the \emph{worst-case} edge of a graph. We show that under this worst-case edge analysis, our choice of $\nu$ is essentially optimal. To obtain better approximation ratios, each angle $\theta_{ij}$ in \cref{alg:epr} must depend on more than just $g_{ij}$, or the analysis must avoid reducing to the worst-case edge.

\section{Semidefinite relaxation of the EPR problem}\label{sec:sos}
To prove the approximation ratio, we find two numbers $l(G), u(G)$ that depend on the input graph, such that $u(G) \ge \lambda_{\max}(H(G))\ge \langle \psi _G|H(G)|\psi_G\rangle \ge l(G) \ge 0$. We show $\ell(G)\ge \alpha\cdot u(G)$, which in turn gives $\langle \psi _G |H(G)|\psi_G\rangle \ge \alpha \cdot\lambda_{\max}\of{H(G)}$, proving the approximation ratio. 
For the lower bound, we use \cref{lem:king_simplification}. We construct $u$ by upper-bounding a relaxation of the EPR problem given by the \textit{quantum moment-SoS hierarchy} \cite{navascues2008}, which we briefly introduce here. For more detailed descriptions, see \cite{gharibian2019, king2023, marwaha2025}. The idea of using a semidefinite program to upper bound an objective function on graphs goes back to the work of Goemans and Williamson on MaxCut problem \cite{goemans1995}.

\subsection{Defining the relaxation}\label{sec:sos/def}
Consider the Pauli monomials on $n$ qubits with at most $k$ non-identity terms:
\begin{align}
\label{eqn_mathcalPK}
    \mathcal{P}_k \defeq \Big\{ \sigma_{i_1}^{\alpha_1} \cdots \sigma_{i_t}^{\alpha_t} \,\big|\, t \leq k, \, \alpha_j \in \{X, Y, Z\}, \,  1 \leq i_1 < \dots < i_t \leq n \Big\}\,,
\end{align}
as well as their span with respect to real coefficients:
\begin{align*}
    \mathcal{O}_k \defeq \mathrm{span}_{\mathbb{R}}\Big\{ \mathcal{P}_k \Big\}\,.
\end{align*}
The $k$th level of the quantum moment-SoS hierarchy is defined with respect to $\mathcal{M}_k$, which is the set of real symmetric \emph{moment matrices}, $\Gamma \in \mathbb{R}^{\mathcal{P}_k \times \mathcal{P}_k}$ (the notation $\mathbb{R}^S$ refers to a real vector indexed by elements of set $S$) satisfying:
\begin{align}
& \Gamma \succeq 0, \label{eq:epr_psd}\\
& \Gamma(A,B) = \Gamma(A',B') &&\forall\, A,B,A',B' \in \mathcal{P}_k: AB = A'B', \label{qwqwpero}\\
& \Gamma(A,B) = -\Gamma(A',B') &&\forall\, A,B,A',B' \in \mathcal{P}_k : AB = -A'B', \label{qwqwpero2}\\
& \Gamma(A,B) = 0 &&\forall\,A,B \in \mathcal{P}_k: AB \text{ not Hermitian}, \nonumber\\
& \Gamma(A,A) = 1 &&\forall\, A \in \mathcal{P}_k. \nonumber
\end{align}
For convenience, we define the real linear functional 
$L : \mathcal{O}_{2k} \to \mathbb{R}$ that satisfies  $L(C) = \Gamma(A,B)$ whenever $C = AB$  with $A, B \in \mathcal{P}_k$. The equality constraints of $\mathcal{M}_k$ ensure that $L(C)$ is well defined, and \cref{eq:epr_psd} is equivalent to $L(A^2) \geq 0$ for all $A \in \mathcal{O}_k$ (see \cite[Lemma 1.44]{burgdorf2016optimization}). The value $L(C)$ is also called the \emph{pseudo-expectation} of the operator $C \in \mathcal{O}_{2k}$.

The $k$th level of the quantum moment-SoS hierarchy is then given by the following semidefinite program (SDP)
\begin{equation}
\label{eq:epr_sdp}
\begin{aligned}
\max \quad & L(H(G)),  \\
\text{s.t.} \quad 
& \Gamma \in \mathcal{M}_k. 
\end{aligned}
\end{equation}
From the output of this SDP, we obtain the values
\begin{align}
    g_{ij} \defeq \frac{-1+L(X_iX_j)-L(Y_iY_j)+L(Z_iZ_j)}{2}, \quad g_{ij}^+ \defeq \max\{g_{ij},0\}\,, \label{eq:sos_g_values}
\end{align}
where $1 + g_{ij}$ is the relaxed objective value on edge $(i,j)$. It then holds that 
\begin{align}
    u(G) \defeq \sum_{(i,j) \in E} w_{ij} (1 + g_{ij}), \label{eq:upper_bound}
\end{align}
is an upper bound for $\lambdaMax{H(G)}$.

\subsection{Monogamy of entanglement}\label{sec:sos/moe}

A key element for designing approximation algorithms for the EPR problem is \emph{monogamy of entanglement} (MoE) \cite{anshu2020, parekh2021a, king2023}. These statements bound the sum of energies (from either the original problem or its relaxation) on neighboring edges. Most research concerns MoE for the Quantum MaxCut (QMC) problem \cite{anshu2020, parekh2021a, parekh2022}, which is distinct from the EPR problem. Specifically, given an edge weighted graph $G= (V,E,w)$, the QMC problem is to find the maximum eigenvalue of the Hamiltonian
\newcommand{\qmcHamilt}{H^{\mathrm{QMC}}(G)}
\newcommand{\qmcLocHamilt}{h_{ij}^{\mathrm{QMC}}}
\begin{align*}
    \qmcHamilt \defeq \sum_{(i,j)\in E} w_{ij} \qmcLocHamilt{}, \text{ for } \qmcLocHamilt \defeq \frac{1}{2}\left( I_i I_j - X_i X_j - Y_i Y_j - Z_i Z_j \right),
\end{align*}
similar to \cref{eqn_eprHamiltDef}. For bipartite graphs, the QMC and EPR problems are equivalent \cite{king2023}. Here, we show that on bipartite graphs, the SDP relaxations of QMC and EPR are also equivalent. To show this, note that
\begin{align}
\label{iqjwper}
    q_{ij} \defeq \frac{-1-L(X_iX_j)-L(Y_iY_j)-L(Z_iZ_j)}{2}
\end{align}
is such that $1+q_{ij}$ is the relaxed objective value on edge $(i,j)$ for the QMC problem. That is, $q_{ij}$ is the QMC analogue of $g_{ij}$, see \cref{eq:sos_g_values}.

\begin{lemma}\label{lem:bipartite_equivalence}
   Let $E$ be the edge set of a bipartite graph and let $k \in \mathbb{N}$. For $g$ and $q$ as in \cref{eq:sos_g_values} and \cref{iqjwper} respectively, define $g(\Gamma)$ and $q(\Gamma)$ as the $g$ and $q$ values induced by some $\Gamma \in \mathcal{M}_k$. We have that
   \newcommand{\setEqGQ}[1]{\setFunct{ \left\{ {#1}(\Gamma)_{ij} \right\}_{(i,j) \in E} }{ \Gamma \in \mathcal{M}_k}}
   \begin{align}
        \setEqGQ{g} = \setEqGQ{q}. 
   \end{align}
\end{lemma}
The proof is deferred to \cref{apx:omitted_proofs/bipartite_equivalence}. \cref{lem:bipartite_equivalence} allows us to take existing MoE bounds for the QMC problem and apply them to the EPR problem. 
For instance, \cref{lem:degree_moe} provides a new MoE bound for the quantum moment-SoS hierarchy for the EPR problem. We show the following simpler corollary of it:
\begin{corollary}\label{cor:moe}
    Fix any graph $G = (V,E,w)$. Then the output $\ofc{g_{ij}}_{(i,j)\in E}\,$ in \cref{eq:sos_g_values} from the $2^{\text{nd}}$ level of the quantum moment-SoS hierarchy 
obeys the following bound for all edges $(i,j) \in E$:
    \begin{align}
    \label{eqn_nonlinearMOE}
        \sum_{k \in N(i) \setminus \ofc{j}} g^+_{ik} &\le Q(g_{ij}^+)\,,
    \end{align}
    where    
    \begin{align}\label{eq:q_def}
        Q(x) &\defeq \begin{cases}
            1-x, \quad &\text{ if } \; 0\le x \le 1/2, \vspace{1ex}\\
            \frac{1}{2}\of{\sqrt{3(1-x^2)}-x}, \quad &\text{ if } \; 1/2 < x \le \sqrt{3}/2, \vspace{1ex} \\
            0,  &\text{ if } \; \sqrt{3}/2 < x \le 1\,.
        \end{cases}
    \end{align}
\end{corollary}
\begin{proof}

If $g_{ij}^+ \le \frac{1}{2}$, \cref{eqn_nonlinearMOE} follows from \cite{parekh2021a} (see also \cite[Lemma 1]{lee2024}) together with \cref{lem:bipartite_equivalence}. Thus, we assume that $g_{ij}^+ > \frac{1}{2}$.

Let $P$ be the RHS of \cref{eq:degree_moe}, i.e. $ \sum_{k \in N(i) \setminus \ofc{j}} g_{ik} \leq P(g_{ij},d_i)$. We use $N^+$ to describe the subset of $N(i)\setminus \{j\}$ with positive values of $g$; i.e. $N^+ \defeq \setFunct{k \in N(i) \setminus \{j\}}{ g_{ik} > 0}$. 
\begin{itemize}
    \item Suppose $|N^+| = 0$. Then all values of $g$ are non-positive. Since $Q$ is non-negative,
    $$
    \sum_{k \in N(i) \setminus \ofc{j}} g_{ik} \le \sum_{k \in N(i) \setminus \ofc{j}} g^+_{ik} = 0 \leq Q(g_{ij}^+)\,.
    $$
    \item Otherwise, $|N^+| \ge 1$. We apply \Cref{lem:degree_moe} to a star graph centered at $i$ that is a subgraph of $G$, where $i$ is adjacent to $j$ and to all $k \in N^+$:
    $$
    \sum_{k \in N(i) \setminus \ofc{j}} g^+_{ik} 
    = \sum_{k \in N^+} g_{ik} 
    \leq 
    P(g_{ij}, |N^+| + 1)
    \leq 
    \max_{d \in \mathbb{N}, d \geq 2} P(g_{ij}, d)\,.
    $$
   Recall that we assumed $g_{ij}^+ \ge \frac{1}{2}$. When $d \ge 2$ and $x \ge \frac{1}{2}$, 
        $$
        \frac{\partial}{\partial d} P(x,d) = \frac{1}{2} \left(-1 + \frac{d\sqrt{1-x^2}}{\sqrt{d^2-1}} \right) \leq \frac{1}{2}\left(-1 + \frac{\sqrt{3}}{2} \cdot \frac{d}{\sqrt{d^2-1}} \right)\,,
        $$
        which is $\le 0$ when $d \ge 2$. So in this case $\max_{d \in \mathbb{N}, d \geq 2} P(g_{ij}, d) = P(g_{ij}, 2) \le Q(g_{ij}^+)$.
    \qedhere
\end{itemize}
\end{proof}

Furthermore, due to the equivalence of the optimal values of the moment-SoS relaxations for the EPR problem and Quantum MaxCut (QMC) on bipartite graphs, our bounds also apply when $g$ instead refers to the SDP edge value for the moment-SoS relaxation of QMC.

\section{Analysis}\label{sec:analysis}
We use this section to prove \Cref{claim:all_cases}, which lower-bounds the approximation ratio $\alpha$ of \Cref{alg:epr} depending on some parameters. We then choose explicit parameters in \Cref{thm:apx_ratio_formal} that gives $\alpha > 0.8395$, and prove this in \Cref{apx:proof_main_result}. For convenience, we provide a list of important variables and functions, and their uses at the end of this document in \Cref{tab:notation}.

Consider $\nu$ of the following form,
given $\beta > \frac{1}{2}$, function $\Theta \in \mathcal{C}$, and function $\Lambda: [0,1] \to [0,1]$:
\begin{align}
\label{eq:parameterization}
    \nu(x) \defeq   \arcsin \sqrt{ \tilde{\nu}(x)}\,, \quad \quad \tilde{\nu}(x) \defeq \begin{cases}
        \Theta(x^+), & \text{if}\; x \le \beta\,, \\
       \Lambda(x^+), & \text{if} \;x > \beta\,.
    \end{cases}
\end{align}
We analyze the approximation ratio achieved by \Cref{alg:epr} for this choice of $\nu$. As described in the introduction, we will show the algorithm is an $\alpha$-approximation by proving  
$$
\bra{\psi_G} H(G) \ket{\psi_G} \ge \ell(G) \ge \alpha \cdot u(G) \ge \alpha \cdot \lambdaMax{H(G)}\,.
$$
We use $\ell$ from \Cref{lem:king_simplification} and $u$ from \cref{eq:upper_bound}. Expanding these expressions, we get
\begin{align}\label{eq:alpha_ratio}
    \alpha \defeq \min_{G} \frac{\ell(G)}{u(G)}
    &= 
    \min_G \frac{
    \sum_{(i,j) \in E} \frac{w_{ij}}{2}  \big(1 + A_{ij} A_{ji} + (A_{ij} + A_{ji}) \sin \nu(g_{ij})\big)
    }
    {
    \sum_{(i,j) \in E} w_{ij} \, (1 + g_{ij})
    }
    \\
    &\ge 
    \min_G \min_{\substack{(i,j) \in E\\ 1+g_{ij}>0}} 
    \frac{1 + A_{ij} A_{ji} + (A_{ij} + A_{ji}) \sin \nu(g_{ij})}{2(1 + g_{ij})}\,. \label{eq:alphaRatio2}
\end{align}
The right-hand side expression only depends on the values $\{g_{ij}\}_{(i,j) \in E}$, which obey \Cref{cor:moe}. In fact, the expression only depends on $g_{k\ell}$ incident to $i$ or $j$. We thus minimize this expression over values $\{g_{ij}\} \cup \{g_{ik}\}_{k \in K_i} \cup \{g_{kj}\}_{k \in K_j}$ obeying \Cref{cor:moe}, given nodes $i$ and $j$, and arbitrary-size sets of ``other'' neighbors $K_i \defeq N(i) \setminus \{j\}$ and $K_j \defeq N(j) \setminus \{i\}$.

Using our parameterization of $\nu$ in \cref{eq:parameterization}, we can simplify some of the above expressions:
\begin{align*}
    \sin{\nu(g_{ij})} &= \sqrt{\tilde{\nu}(g_{ij}^+)}\,, 
    \quad\quad A_{ij}  = \prod_{k \in K_i} \sqrt{1-\tilde{\nu}(g_{ik}^+)}\,.
\end{align*}
The reason to use the set $\mathcal{C}$ from \Cref{def:theta_set} is demonstrated by the following two lemmas:

\begin{lemma}\label{lem:all_small_a_bound}
    Suppose $g_{ik} \le \beta$ for all $k \in K_i$. Then $A_{ij} \ge \sqrt{1-\Theta\of{Q(g_{ij}^+)}}$.
\end{lemma}
\begin{proof}
    Observe that
    \begin{align*}
        A_{ij} &= \prod_{k \in K_i} \sqrt{1-\tilde{\nu}(g_{ik}^+)} = \prod_{k \in K_i} \sqrt{1-\Theta(g_{ik}^+)} \ge \sqrt{1-\Theta\Big(\sum_{k \in K_i}g_{ik}^+}\Big) \ge \sqrt{1-\Theta\of{Q(g_{ij}^+)}}\,.
    \end{align*}
    The equalities hold by \cref{eq:parameterization}, the first inequality holds by \cref{def:theta_set}, and the last inequality holds by \cref{cor:moe}.
\end{proof}

\begin{lemma}\label{lem:one_big_a_bound}
    Suppose $g_{ik'} > \beta$ for some $k' \in K_i$. Then $A_{ij} \ge f(g_{ij}, g_{ik'})$, where $f: [-1,1] \times [-1,1] \to \mathbb{R}$ is the function
    $$
    f(x,y) \defeq \sqrt{\big(1-\Lambda(y^+)\big) \big(1-\Theta\left( Q(y^+)-x^+\right)\big)}\,.
    $$
\end{lemma}
\begin{proof}
    Let $K_i^- \defeq K_i \setminus \ofc{k'}$. Since $g_{ik'} > \beta$, \Cref{cor:moe} implies $g_{ik} \le (1-\beta)$ for every $k \in K_i^-$. Since $\beta > \frac{1}{2}$, $1-\beta < \beta$, and so $g_{ik} < \beta$ for all $k \in K_i^-$. We then observe that
    \begin{align*}
        A_{ij} = \prod_{k \in K_i} \sqrt{1-\tilde{\nu}(g_{ik}^+)} 
        = \sqrt{1-\Lambda(g_{ik'}^+)} \  \cdot \prod_{k \in K_i^-} \sqrt{1-\Theta(g_{ik}^+)} 
        \ge \sqrt{1-\Lambda(g_{ik'}^+)} \cdot \sqrt{1-\Theta\Big(\sum_{k \in K_i^-} g_{ik}^+\Big)} \,.
    \end{align*}
    The equalities hold by \cref{eq:parameterization} and the  inequality holds by \Cref{def:theta_set}. 
    By \Cref{cor:moe}, the sum $\sum_{k \in K_i^-} g_{ik}^+\le  Q(g_{ik'}^+) -g_{ik'}^+$. By \Cref{def:theta_set}, $\Theta$ is monotonically increasing.
\end{proof}
We use \Cref{lem:all_small_a_bound,lem:one_big_a_bound} in the following case-wise analysis of the ratio in \cref{eq:alpha_ratio}:
\paragraph{Case 1:}\!\! Suppose $g_{ij} \le \beta$ on edge $(i,j)$ and on all neighboring edges.
In this case, we may apply \cref{lem:all_small_a_bound} on both $A_{ij}$ and $A_{ji}$, to derive that the ratio \cref{eq:alpha_ratio} is lower bounded by
\begin{align}
\label{eqn:case1}
    r_1(g_{ij}), \text{ for } r_1(g) \defeq \frac{2-\Theta\big(Q(g^+)\big)+2\sqrt{\Theta(g^+) \of{1-\Theta\big(Q(g^+)\big)}}}{2\of{1+g}}\,.
\end{align}
When $g_{ij}\le 0$, the numerator of $r_1(g)$ is constant, but the denominator increases with $g$. Therefore, it follows that
\begin{align}
\label{iojpappp}
    r_1(g_{ij}) \geq \min_{-1 < g \leq \beta} r_1(g) = \min_{0 \leq g \leq \beta} r_1(g).
\end{align}

\paragraph{Case 2:}\!\! Suppose $g_{ij} > \beta$.
In this case, by \cref{cor:moe}, we have that $g_{ik}\le \beta$ for all $k \in K_i$ and $g_{kj} \le \beta$ for all $k \in K_j$. Thus, we may again apply \cref{lem:all_small_a_bound} to both $A_{ij}$ and $A_{ji}$. The only difference is that because $g_{ij} \ge \beta$, we have $\tilde{\nu}(g_{ij}^+) = \Lambda(g_{ij}^+)$ by \cref{eq:parameterization}. So the ratio \cref{eq:alpha_ratio} is lower bounded by
\begin{align}
\label{eqn:case2}
r_2(g_{ij}), \text{ for } r_2(g) \defeq
    \frac{2-\Theta\big(Q(g^+)\big)+2\sqrt{\Lambda(g^+) \of{1-\Theta\big(Q(g^+)\big)}}}{2\of{1+g}}\,.
\end{align}
Since $g_{ij} \in [\beta,1]$, we have that $r_2(g_{ij}) \geq \min_{\beta \leq g \leq 1} r_2(g)$.

\paragraph{Case 3:}\!\! 
Suppose $g_{ik'} > \beta$ for some $k' \in K_i$. In this case, all other edges incident to $i$ must have $g_{ij} \le 1-\beta\le \beta$ by \cref{cor:moe}. We can then apply \cref{lem:one_big_a_bound} to $A_{ij}$ to obtain the following lower bound on \cref{eq:alpha_ratio}:
\begin{align}
\label{eqn:case3}
    \frac{1+ f(g_{ij}, g_{ik'}) \cdot A_{ji} + \sqrt{\Theta(g_{ij}^+)}\of{f(g_{ij}, g_{ik'}) + A_{ji}}}{2\of{1+g_{ij}}}\,.
\end{align}
There may or may not be some $\ell \in K_j$ with $g_{\ell j} > \beta$.
As such, we split into two subcases:

\textit{Case 3a:} Exactly one $g_{\ell' j} > \beta$. Then by \cref{lem:one_big_a_bound},  $A_{ji} \ge f(g_{ij}, g_{\ell' j})$.

\textit{Case 3b:} All $g_{\ell j} \le \beta$. Then by \cref{lem:all_small_a_bound}, $A_{ji} \geq \sqrt{1-\Theta\big(Q(g_{ij}^+)\big)}$.

So far, \cref{eqn:case3} for Case 3a depends on three variables $g_{ij}$, $g_{ik'}$ and $g_{\ell' j}$. For Case 3b, \cref{eqn:case3} depends on two variables $g_{ij}$ and $g_{ik'}$. We can relate the variables using the 
monogamy of entanglement claim from \cite{lee2024}:
\begin{claim}[\cite{lee2024}]
\label{claim:moe_lp24}
Fix any graph $G = (V,E,w)$. Then the output $\ofc{g_{ij}}_{(i,j)\in E}\,$ in \cref{eq:sos_g_values} from the $2$nd level (and higher levels) of the quantum moment-SoS hierarchy obeys $g_{ij} \le R(g_{ik})$ and $g_{ik} \le R(g_{ij})$ for all pairs of neighboring edges $\{(i,j), (i,k) \} \subseteq E$, where
\begin{align}\label{eq:r_def}
R(x) \defeq \frac{1}{2} \left(\sqrt{3(1-x^2)} - x \right)\,.
\end{align}
\end{claim}
Note that $R$ is monotonically decreasing for $x \ge \frac{1}{2}$. Together with \Cref{claim:moe_lp24}, we observe
    $$
g_{ij} \le \min\left\{R(g_{ik'}), R(g_{\ell' j}) \right\} \le R(\beta)\,, 
\quad \quad
\max\{g_{ik'},g_{\ell' j}\} \le R(g_{ij})\,.
$$
We now assume a property of $\Lambda$ to remove $g_{ik'}$ and $g_{\ell' j}$ from the minimization problem altogether. 
\begin{definition}\label{def:lambda_set}
Fix some $1/2 < \beta \le 1$ and $\Theta \in \mathcal{C}$. Let $\mathcal{D}_{\Theta, \beta}$ be the set of functions \begin{align*}
        \Lambda\,:\, \ofb{0,1} \rightarrow \ofb{0,1},
    \end{align*}
    where  $f$ from \cref{lem:one_big_a_bound} satisfies $f(x,y) \ge f^*(x)$ for all $\beta \le y \le R(x)$ and $-1 \le x \le R(\beta)$, and
    \begin{align}
f^*(x) \defeq f(x,R(x)) =  \sqrt{\bigl(1-\Lambda(R(x)^+)\bigr)\cdot
          \bigl(1-\Theta\left(Q(R(x)^+)-x^+
          \right)\bigr)}\,.
\end{align}
\end{definition}
If $\Lambda \in \mathcal{D}_{\Theta, \beta}$, then in Case 3 of \Cref{lem:one_big_a_bound}, we can lower-bound both $f(g_{ij}, g_{ik'})$ and $f(g_{ij}, g_{\ell'j})$ with $f(g_{ij}, R(g_{ij}))$, since $g_{ik'}, g_{\ell' j} \le R(g_{ij})$. Thus, if $\Lambda \in \mathcal{D}_{\Theta, \beta}$, 
we can without loss of generality minimize \cref{eqn:case3} over $-1 \le g_{ij} \le R(\beta)$ with fixed $g_{ik'}=R(g_{ij})$ and fixed $g_{j\ell'}=R(g_{ij})$. Thus, in Case 3a, \cref{eqn:case3} is lower bounded by
\begin{align}
\label{eqn_r3Def}
    r_3(g_{ij}), \text{ for } r_3(g) \defeq \frac{\,1 + {f^*(g)}^2 \;+\; 2\sqrt{\Theta(g^+)}f^*(g)}{2(1+g)},
\end{align}
and $r_3(g_{ij}) \geq \min_{-1 < g \leq R(\beta)} r_3(g)$. In Case 3b, we find that \cref{eqn:case3} is lower bounded by
\begin{align}
    &\frac{\,1 + f^*(g_{ij})\,\sqrt{1-\Theta(Q(g_{ij}^+))} \;+\; \sqrt{\Theta(g_{ij}^+)}\left(f^*(g_{ij})+\sqrt{1-\Theta(Q(g_{ij}^+))}\, \right)}{2(1+g_{ij})} \nonumber
    \\ &\geq \min\{ r_1(g_{ij}), r_3(g_{ij}) \} \geq  \min\left\{ \min_{0 \leq g \leq R(\beta)} r_1(g),  \min_{-1 < g \leq R(\beta) }r_3(g) \right\}. \label{qwerjipqwp}
\end{align}
The first inequality in \cref{qwerjipqwp} is due to the fact that $$f^*(g), \sqrt{1-\Theta(Q(g^+))} \geq \min \left\{ f^*(g),\sqrt{1-\Theta(Q(g^+))} \right\}  \quad \forall g \in \mathbb{R}.$$ The second inequality in \cref{qwerjipqwp} is due to \cref{iojpappp}. Note also that $R(\beta) \leq \beta$, so that \linebreak $\min_{0 \leq g \leq R(\beta)} r_1(g) \geq \min_{0 \leq g \leq \beta} r_1(g)$.

We lower bound the approximation ratio by combining all cases (\cref{eqn:case1,eqn:case2,eqn:case3}):
\begin{claim}
\label{claim:all_cases}
    Suppose $\nu$ has the form \cref{eq:parameterization} for some $\beta > \frac{1}{2}$, function $\Theta \in \mathcal{C}$, and function $\Lambda \in \mathcal{D}_{\Theta, \beta}$. Then \Cref{alg:epr} with this choice of $\nu$ has approximation ratio at least 
\begin{align} 
    \alpha \;\ge\; \min\Bigg\{ 
    &\underset{0\le g\le \beta}{\min}\; 
    r_1(g)\,,\;
    \underset{\beta\le g\le 1}{\min}\;
    r_2(g)\,,\;
    \underset{-1  < g\le R(\beta)}{\min}\;
    r_3(g)\,
    \Bigg\}\label{eq:minimization_problem}\,
\end{align}
for functions $r_1$, $r_2$ and $r_3$ as defined in \cref{eqn:case1,eqn:case2,eqn_r3Def} respectively. 
\end{claim}
It remains to choose $\beta$, $\Theta \in \mathcal{C}$, and $\Lambda \in \mathcal{D}_{\Theta, \beta}$ that obtain a large value of \cref{eq:minimization_problem}.
We provide such a choice in the following theorem. 
\begin{theorem}[formal restatement of \cref{thm:apx_ratio}]\label{thm:apx_ratio_formal}
The value of the minimization problem \cref{eq:minimization_problem} is at least $\alpha \ge \alpha' \defeq 0.839511$ with the following choice of parameters:
\begin{itemize}
    \item $\nu$ is the function \vspace{-2ex}
    \begin{align*}
            \nu(x) \defeq   \arcsin \sqrt{ \tilde{\nu}(x)}\,, \quad \quad \tilde{\nu}(x) \defeq \begin{cases}
        \Theta(x^+), & \text{if}\; x \le \beta\,, \\
       \Lambda(x^+), & \text{if} \;x > \beta\,,
       \end{cases}
       \vspace{-2ex}
    \end{align*}
    \item $\beta=0.67, \;  \gamma = 0.049$,
    \item $\Theta$ is the piecewise linear function defined by the points
    \begin{align}\label{eq:theta_points}
        \ofc{\of{0,0}, \of{Q(\beta),\gamma}, \of{\beta,  \frac{(\gamma/2 + \alpha'(1+\beta)-1)^2 }{1-\gamma}\approx 0.1913}, \of{1, 2(1-\alpha')\approx 0.3210}}.
    \end{align}
    \item $\Lambda$ is the function \begin{align}\label{eq:lambda_choice}
        \Lambda(x) = \frac{\of{\frac{1}{2}\Theta\of{Q(x)}+\alpha'(1+x)-1}^2}{1-\Theta\of{Q(x)}}.
        \end{align}
\end{itemize}
As a consequence, \Cref{alg:epr} with this choice of $\nu$ has approximation ratio $\alpha \ge 0.839511$.
\end{theorem}
We provide some intuition for the choice of parameters. $\Lambda$ is chosen such that $r_2(g)$ is exactly $\alpha'$ in the entire domain. $\Theta$ is chosen such that $r_1(g)$ is exactly $\alpha'$ at points $0$ and $\beta$. We then choose $\beta$ and $\gamma$ such that $\Lambda \in \mathcal{D}$. We plot functions $r_1$ through $r_3$ in \cref{fig:r_functions}. We also provide an interactive online plot 
\href{https://www.desmos.com/calculator/d0bken19do}{here}.
\begin{figure}
    \centering
    \includegraphics[width=.75\linewidth]{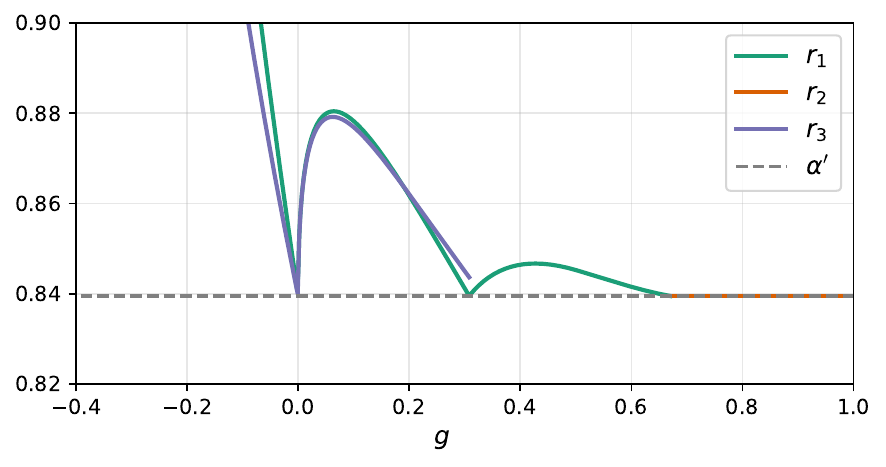}
    \caption{Values of the functions $r_1$, $r_2$ and $r_3$ as defined in \cref{eqn:case1,eqn:case2,eqn_r3Def} respectively. }
    \label{fig:r_functions}
\end{figure}
We defer the proof of \Cref{thm:apx_ratio_formal} to \cref{apx:proof_main_result}.

\section{Limits on approximability}\label{sec:limits}

We now provide some upper bounds on the approximation that can be achieved by the methods outlined in our work.

\subsection{Limitations from analysis}

Our first limitation comes from our assumptions in the analysis of \Cref{alg:epr}. Specifically, we assume that $\theta_{ij}$ depends \textit{only} on $g_{ij}$, and that we take the \emph{worst-case} ratio over edges. With mild assumptions, \textit{any} such analysis obtains an $\alpha$-approximation where $\alpha < \ratioUB{}$. (Recall that \Cref{thm:apx_ratio_formal} shows the existence of $\alpha \ge 0.839511$.)
\begin{lemma}\label{lem:limit/analysis}
    Using a worst-case edge analysis, \cref{alg:epr} with any choice of $\nu$ where $\nu(0) = 0$ is at most a $\ratioUB{}$-approximation for the EPR problem.
\end{lemma}
\begin{proof}
Consider the four-node path graph $P_4$. 
It is known that the SDP \cref{eq:epr_sdp} returns
$(g_{12}, g_{23}, g_{34}) = (\sqrt{3}/2, 0, \sqrt{3}/2)$~\cite[Section 4.2.2]{takahashi2023}. We analyze the approximation ratio of \Cref{alg:epr} on this graph using \Cref{lem:king_simplification} (which is tight on triangle-free graphs) and a worst-case edge analysis. Hence, the approximation ratio $\alpha$ satisfies
\begin{align*}
    \alpha &\le \min\left\{  \frac{\bra{\psi_G} h_{23}  \ket{\psi_G} }{1+g_{23}}, \frac{ \bra{\psi_G} h_{12}  \ket{\psi_G} }{1+g_{12}}\right\} \nonumber \\
    &\leq \min \left\{
    \frac{1 + \cos^2\theta_s + 2 \cos\theta_s \sin\theta_m }{2},
    \frac{1 + \cos \theta_m + (1 + \cos \theta_m) \sin\theta_s }{2 + \sqrt{3}}
    \right\}\,,
\end{align*}
where $\theta_s = \nu(\sqrt{3}/2)$ and $\theta_m = \nu(0)$. If we assume $\theta_m = \nu(0) = 0$, then 
\begin{align*}
    \alpha \le \max_{-1 \leq x \leq 1 } \min\left\{ \frac{2-x^2}{2}, \frac{2+2x}{2+\sqrt{3}}\right\} = \max_{0 \leq x \leq 1 } \min\left\{ \frac{2-x^2}{2}, \frac{2+2x}{2+\sqrt{3}}\right\}\,,
\end{align*}  
where $x = \sin{\theta_s}$. Since $ \frac{2-x^2}{2}$ is a concave parabola and $\frac{2+2x}{2+\sqrt{3}}$ is a positively sloped line which intersect exactly once in the interval $[0,1]$ we deduce that $\alpha$ is given by the common value at the point of intersection. Solving the corresponding quadratic equation 
\begin{align*}
    \frac{2-x^2}{2}\!=\!\frac{2+2x}{2+\sqrt{3}},~ 0\leq x\leq 1 \!\implies\!\! x\!=\!\frac{-2 + \sqrt{10 + 4 \sqrt{3}}}{2 + \sqrt{3}}\!,~\alpha\!=\!\frac{2\bigl(\sqrt{3}+\sqrt{10+4\sqrt{3}}\bigr)}{(2+\sqrt{3})^{2}}\!\simeq\! 0.8395111. \quad\quad\qedhere
\end{align*}
\end{proof}
We do not view $\nu(0) = 0$ as restrictive. In fact, if $\nu(\epsilon) \nrightarrow 0$ as $\epsilon \to 0^+$, then \Cref{alg:epr} only achieves an $0.5$-approximation on the complete bipartite graph $K_{a,a}$ for large values of $a$ ~\cite{takahashi2023}.

\subsection{Limitations from ansatz}\label{sec:limits/ansatz}
The ansatz we use in this work (and used in the initial algorithm of \cite{king2023}) has a natural upper bound of $0.873$. This was suggested by \cite{tao2025}; here we give a formal proof. 
\begin{lemma}\label{lem:limit/ansatz}
     \Cref{alg:epr} achieves at most a $\frac{3+\sqrt{5}}{6}\approx 0.8727$-approximation on the EPR problem, even if each angle $\theta_{ij}$ could depend on the entire SDP output $\{g_{ij}\}_{(i,j) \in E}$.
\end{lemma}
\begin{proof}
Consider the unweighted cycle $C_4$ defined by
    \begin{align*}   
        V = \{1,2,3,4\}, \quad E = \{(1,2),(2,3),(3,4),(4,1)\}\,.
    \end{align*}
    Let $\{a,b,c,d\}=\{\theta_{12},\theta_{23},\theta_{34},\theta_{41}\}$ denote the four angles in \Cref{alg:epr}.
    Using \Cref{lem:king_simplification}, we have 
    \begin{align*}
        2\braket{\chi|  H(C_4) | \chi} = \;&1+\cos{d}\cos{b}+\left(\cos{d}+\cos{b}\right) \sin{a}+1+\cos{a}\cos{c}+\left(\cos{a}+\cos{c}\right)\sin{b} \\
        +&1+\cos{b}\cos{d}+\left(\cos{b}+\cos{d}\right) \sin{c}
        +1+\cos{c}\cos{a}+\left(\cos{c}+\cos{a}\right) \sin{d},
    \end{align*}
    where each set of three terms corresponds to a single edge.
    Since $C_4$ has no triangles, note that the equation in \Cref{lem:king_simplification} is an equality.  
    To upper bound $\braket{\chi|  H(C_4) | \chi}$, we rewrite each term of the form $\sin u \cos v$ as
    \begin{align*}
        \sin u\,\cos v
        = 
        \left(\frac1{\sqrt{x}}\sin u\right)\left(\sqrt{x}\cos v\right),
    \end{align*}
    for some arbitrary scalar $x>0$. Then, defining the vectors 
    \begin{align*}
    \mathbf{a} &= 
    \left(\cos{d}, \cos{a}, \cos{b}, \cos{c}\right) || 
    \frac{1}{\sqrt{x}}\left(\sin{a},\sin{b},\sin{c},\sin{d}\right) || \sqrt{x}\left(\cos{b}, \cos{c}, \cos{d}, \cos{a}\right), \\
    \mathbf{b} &= \left(\cos{b}, \cos{c}, \cos{d}, \cos{a}\right) || \sqrt{x}\left(\cos{d}, \cos{a}, \cos{b}, \cos{c}\right) || 
    \frac{1}{\sqrt{x}}\left(\sin{a},\sin{b}, \sin{c}, \sin{d}\right),
    \end{align*}
    where $||$ denotes concatenation allows us to express
    \begin{align*}
        2  \braket{\chi|  H(C_4) | \chi} &= 4 + \mathbf{a} \cdot \mathbf{b} \le 4+\lVert\mathbf{a}\rVert_2\,\lVert\mathbf{b}\rVert_2= 4+\| \mathbf{a} \|_2^2 \\
    &=4+(1+x)\left(\cos^2a \!+\! \cos^2b \!+\! \cos^2c \!+\! \cos^2d\right) +\frac{1}{x}\left(\sin^2a \!+\! \sin^2b \!+\! \sin^2c \!+\! \sin^2d\right),
    \end{align*}
    where in the first line we apply the Cauchy--Schwarz inequality. Choosing $x$ such that $1+x=\tfrac{1}{x}$ yields $x=(\sqrt5-1)/2$. Using $\cos^2\theta + \sin^2\theta =1$ then yields
    \begin{align}
    \label{aisodfs}
      \braket{\chi|  H(C_4) | \chi} \le  3 + \sqrt{5}\,.
    \end{align}
    In fact, by picking $a=b=c=d=\tan^{-1}{x}$ one obtains that $\braket{\chi|  H(C_4) | \chi} =  3 + \sqrt{5}$. We know that $\lambdaMax{H(C_4)} = 6$ because $C_4$ is complete and bipartite~\cite{lieb1962,takahashi2023}. 
    So this algorithm is an $\alpha$-approximation of \textit{at most} $\alpha \le \frac{3+\sqrt{5}}{6}$ with a tight upper bound. 
\end{proof}
Several works for the EPR problem~\cite{king2023, apte2025, apte2025b, tao2025} use this ansatz in \cref{eq:epr_chi}. \Cref{lem:limit/ansatz} demonstrates that a new approach is required to boost $\alpha > 0.873$.

\section*{Acknowledgments}
E.L. is supported by a KIAS Individual Grant CG093802 at Korea Institute for
Advanced Study.
K.M. and J.S. acknowledge that this material is based upon work supported by the National Science Foundation Graduate Research Fellowship under Grant No.\ 2140001. 
K.M. acknowledges support from AFOSR (FA9550-21-1-0008). 
O.P. acknowledges that this material is based upon work supported by the U.S. Department of Energy, Office of Science, Accelerated Research in Quantum Computing, Fundamental Algorithmic Research toward Quantum Utility (FAR-Qu). 
J.S. acknowledges that this work is funded in part by the STAQ project under award NSF Phy-232580; in part by the US Department of Energy Office of Advanced Scientific Computing Research, Accelerated Research for Quantum Computing Program.

This article has been authored by an employee of National Technology \& Engineering Solutions of Sandia, LLC under Contract No.\ DE-NA0003525 with the U.S. Department of Energy (DOE). The employee owns all right, title and interest in and to the article and is solely responsible for its contents. The United States Government retains and the publisher, by accepting the article for publication, acknowledges that the United States Government retains a non-exclusive, paid-up, irrevocable, world-wide license to publish or reproduce the published form of this article or allow others to do so, for United States Government purposes. The DOE will provide public access to these results of federally sponsored research in accordance with the DOE Public Access Plan \url{https://www.energy.gov/downloads/doe-public-access-plan}.

This paper was prepared for informational purposes with contribution fro the Global Technology Applied Research center of JPMorgan Chase \& Co. This paper is not a product of the Research Department of JPMorgan Chase \& Co or its affiliates. Neither JPMorgan Chase \& Co nor any of its affiliates makes any explicit or implied representation or warranty and none of them accept any liability in connection with this paper, including, without limitation, with
respect to the completeness, accuracy, or reliability of the information contained herein and the potential legal, compliance, tax, or accounting effects thereof. This document is not intended as investment research or investment advice, or as a recommendation, offer, or solicitation for the purchase or sale of any security, financial instrument, financial product or service, or to be used in any way for evaluating the merits of participating in any
transaction.

\printbibliography

\appendix

\crefalias{section}{appendix}
\crefalias{subsection}{appendix}

\clearpage
\newpage

\section{Omitted proofs}
\label{apx:omitted_proofs}

\subsection{Proof of \texorpdfstring{\cref{lem:bipartite_equivalence}}{QMC and EPR bipartite equivalence}}
\label{apx:omitted_proofs/bipartite_equivalence}
We first prove \cref{lem:bipartite_equivalence}, since \cref{lem:bipartite_equivalence} is used in the proof of \cref{lem:degree_moe} in \cref{section_proofLemma2}.
\begin{proof}
\newcommand{\suppXZ}[1]{\mathrm{supp}_{x,z}\mleft(#1\mright)}
\newcommand{\bipartSet}{\mathcal{V}}
    We show that for any $\Gamma \in \mathcal{M}_k$, there exists a $\tilde{\Gamma} \in \mathcal{M}_k$ that satisfies 
    \begin{align}
    \label{qowiersaf}
     g(\Gamma)_{ij} = q(\tilde{\Gamma})_{ij} \text{ and } q(\Gamma)_{ij} = g(\tilde{\Gamma})_{ij} \, \, \forall (i,j) \in E, 
    \end{align}
    proving the result.

   We write, for any $A \in \mathcal{P}_{k}$, see \cref{eqn_mathcalPK}, $A = \sigma_1 \sigma_2 \cdots \sigma_n$, for $\sigma_i \in \{I, X_i,Y_i,Z_i\}$. We define $\suppXZ{A} \defeq \setFunct{ i \in [n]}{ \sigma_i \in \{ X_i, Z_i\} }$. Let $\bipartSet \subseteq [n]$ correspond to a bipartition of the the graph. That is, for all $(i,j) \in E$, precisely one of $i$ and $j$ is contained in $\bipartSet$. Define $\tilde{\Gamma}$ as the matrix satisfying
    \begin{align}
    \label{qwiper}
        \tilde{\Gamma}(A,B) = C(A,B) \cdot \Gamma(A,B), \text{ for } C(A,B) \defeq (-1)^{| \suppXZ{A} \cap \bipartSet|}  (-1)^{|\suppXZ{B} \cap \bipartSet|}.
    \end{align}
    It is straightforward to verify that $\Gamma$ and $\tilde{\Gamma}$ satisfy \cref{qowiersaf}, and so it remains to show that $\tilde{\Gamma} \in \mathcal{M}_k$. By definition, the matrix $C \defeq C(A,B)_{A,B \in \mathcal{P}_k}$ is a rank 1 PSD matrix. By \cref{qwiper}, $\tilde{\Gamma}$ equals the Hadamard product of the PSD matrices $C$ and $\Gamma$ so that $\tilde{\Gamma} \succeq 0$. It is clear that $\tilde{\Gamma}(A,B) = 0$ whenever $\Gamma(A,B) = 0$, and that $\tilde{\Gamma}(A,A) = 1$ for all $A \in \mathcal{P}_k$.

    \newcommand{\pauliProd}[2]{ \mathcal{K}\mleft(#1,#2\mright)}
    To verify the remaining two conditions of $\mathcal{M}_k$, we define, for any $A,B \in \mathcal{P}_k$ the matrix $\pauliProd{A}{B}$ as the matrix satisfying $\pauliProd{A}{B} \in \mathcal{P}_{2k}$ and $\pauliProd{A}{B} = c AB$ for some $c \in \{ \pm 1, \pm \imagUnit \}$. Define $S_A \defeq \suppXZ{A} \cap \bipartSet$ for any $A \in \mathcal{P}_{2k}$. Note that, for $\triangle$ the symmetric difference operator, we have that $S_A \triangle S_B = \suppXZ{\pauliProd{A}{B}} \cap \bipartSet$. Therefore,
    \begin{align*}
        C(A,B) = (-1)^{|S_A|+|S_B|} = (-1)^{|S_A \triangle S_B|+2|S_A \cap S_B|} = (-1)^{| \suppXZ{\pauliProd{A}{B}} \cap \bipartSet | }.
    \end{align*}
    Now, if $A,B,A',B' \in \mathcal{P}_k$ are such that $AB = A'B'$, or $AB = -A'B'$, then in both cases, $\pauliProd{A}{B} = \pauliProd{A'}{B'}$, so that $C(A,B) = C(A',B')$. Hence $\tilde{\Gamma}$ satisfies \cref{qwqwpero} and \cref{qwqwpero2} and $\tilde{\Gamma} \in \mathcal{M}_k$. 
\end{proof}
\subsection{Proof of \texorpdfstring{\cref{lem:degree_moe}}{New monogamy of entanglement}}\label{apx:omitted_proofs/degree_moe}
\label{section_proofLemma2}
We now prove \cref{lem:degree_moe}, starting with the following claims.

\begin{claim}\label{claim:lem_degree_moe_1}
For $d_i$ and $g_{ij}$ as defined in \cref{eq:degree_moe} of \cref{lem:degree_moe},
\begin{align*}
\frac{1}{2}\of{2-d_i-g_{ij} + \sqrt{(d_i^2-1)(1-g_{ij}^2)}} \leq 1.
\end{align*}
\end{claim}

\begin{proof}
The claim is equivalent to 
\begin{align*}
\sqrt{(d_i^2-1)(1-g_{ij}^2)} \leq d_i + g_{ij},
\end{align*}
which is in turn equivalent to
\begin{align*}
(d_i^2-1)(1-g_{ij}^2) \leq (d_i + g_{ij})^2,
\end{align*}

since $d_i + g_{ij} \geq 0$. Observe that
\begin{align*}
0 \leq (d_i g_{ij} + 1)^2 = (d_i + g_{ij})^2 - (d_i^2-1)(1-g_{ij}^2),
\end{align*}
establishing the claim.
\end{proof}

\begin{claim}\label{claim:lem_degree_moe_2}
\cref{lem:degree_moe} with the assumption that $g_{ij} \in (-1,1)$ and each $g_{ik} > -1$ in \cref{eq:degree_moe} implies the general case.
\end{claim}
\begin{proof}
Fix a vertex $i$ of degree at least 2 and a vertex $j \in N(i)$, and let $D \defeq \ofc{(i,k) \mid k \in N(i)\setminus \ofc{j}}$. Let $\td{D} \defeq \ofc{(i,k) \in D \mid g_{ik} = -1}$ and $\td{d} \defeq |\td{D}|$. 

If $g_{ij} = -1$, then the standard moment-SoS-based star bound \cite{parekh2021a} applied to the edges in $D$ yields \cref{lem:degree_moe}. If $g_{ij}=1$, \cref{lem:degree_moe} follows by applying \cref{claim:moe_lp24} to the pairs $\ofc{(i,j),(i,k)}$ for $(i,k) \in D$. So we need only consider values $g$ with $g_{ij} \in (-1,1)$. 

It suffices to show
\begin{equation}
\label{eq:mod_degree_moe_1}
\begin{aligned}
\sum_{(i,k) \in D \setminus \td{D}} g_{ik} &\le \begin{cases}
        1 + \td{d}, &\text{if } -1 \le g_{ij} < -\frac{1}{d_i}, \\
        \frac{1}{2}\of{2-d_i+2\td{d}-g_{ij} + \sqrt{\of{d_i^2-1}\of{1-g_{ij}^2}}}, &\text{if } -\frac{1}{d_i} \leq g_{ij} \le 1 .
    \end{cases}
\end{aligned}
\end{equation}
Remove the edges in $\td{D}$ from $G$ to produce a graph $\td{G}$, for which \cref{eq:degree_moe} for $i$ and $j$ is equivalent to
\begin{equation}
\label{eq:mod_degree_moe_2}
\begin{aligned}
    \sum_{(i,k) \in D\setminus \td{D}} \hspace{-5px} g_{ik}  &\le \begin{cases}
        1, &\hspace{-7px}\text{if }\! -\!1 \le g_{ij} \!< -\frac{1}{d_i-\td{d}}, \\
        \frac{1}{2}\!\of{2-d_i+\td{d}-g_{ij} + \sqrt{\of{\of{d_i-\td{d}}^2-1}\!\!\of{1-g_{ij}^2}}}, &\hspace{-7px}\text{if }\! -\frac{1}{d_i-\td{d}} \leq g_{ij} \le 1.
    \end{cases}
\end{aligned}
\end{equation}
We will observe that \cref{eq:mod_degree_moe_2} implies \cref{eq:mod_degree_moe_1}. Since each $g_{ik} > -1$ in the sum of \cref{eq:mod_degree_moe_2}, and the hypothesis of the claim must hold for any graph, including $\td{G}$, this will establish the claim. 

Since $ -\frac{1}{d_i-\td{d}} \leq -\frac{1}{d_i}$, if $-\frac{1}{d_i} \leq g_{ij}$ then \cref{eq:mod_degree_moe_2} implies \cref{eq:mod_degree_moe_1} by
\begin{align}\label{eq:moe_bound_comp}
2-d_i+\td{d}-g_{ij} + \sqrt{\of{(d_i-\td{d})^2-1}\of{1-g_{ij}^2}} \leq
2-d_i+2\td{d}-g_{ij} + \sqrt{\of{d_i^2-1}\of{1-g_{ij}^2}}.
\end{align}
If $g_{ij} < -\frac{1}{d_i-\td{d}}$, then we are also fine, leaving the case $-\frac{1}{d_i-\td{d}} \leq g_{ij} < -\frac{1}{d_i}$. For this we seek to show
\begin{align}
\label{eqn_thingWeSeekToShow}
\frac{1}{2}\of{2-d_i+\td{d}-g_{ij} + \sqrt{\of{\of{d_i-\td{d}}^2-1}\of{1-g_{ij}^2}}} \leq 1 + \td{d}.
\end{align}

By using \cref{eq:moe_bound_comp}, we can upper bound the left-hand side of \cref{eqn_thingWeSeekToShow} as follows: 
\begin{align*}
    \frac{1}{2}\of{2-d_i+\td{d}-g_{ij} + \sqrt{((d_i-\td{d})^2-1) (1-g_{ij}^2)}} &\leq \frac{1}{2} \of{2-d_i-g_{ij} + \sqrt{(d_i^2-1)(1-g_{ij}^2)}} + \td{d} \\
    &\leq 1+\td{d},
\end{align*}
where for the last inequality, we have used \cref{claim:lem_degree_moe_1}. This completes the proof.
\end{proof}

\begin{proof}[Proof of \Cref{lem:degree_moe}] 
Fix a vertex $i$ of degree at least 2 and a vertex $j \in N(i)$, and let $D \defeq \ofc{(i,k) \mid k \in N(i)\setminus \ofc{j}}$. We assume $g_{ij} > -1$ and $g_{ik} > -1$ for all $(i,k) \in D$ by \cref{claim:lem_degree_moe_2}.

We will prove the result for the Quantum MaxCut (QMC) values $q$ as defined in \cref{iqjwper}. Since the edge set $D$ is bipartite, we may appeal to \cref{lem:bipartite_equivalence} to obtain the desired result for the EPR values $g$. 

We start by following the approach of the proof of the QMC star bound in \cite[Theorem~4.4]{takahashi2023}. For this we need a moment matrix with respect to the projectors onto a singlet on each edge $(k,l)$: 
\begin{align*}
\Pi_{kl} \defeq \frac{1}{2} h_{kl}^\mathrm{QMC} = \frac{1}{4}( I_k I_l - X_k X_l - Y_k Y_l - Z_k Z_l ).
\end{align*}
Choose a moment matrix $\Gamma \in \mathcal{M}_k$, with $k \geq 2$, that gives rise to the values $g$. 
Since $\mathcal{P}_k$ is a basis for $\mathcal{O}_k$, we extend $\Gamma$ to an operator acting on $\mathcal{O}_k$ by linearity so that $\Gamma(A,B) = L(AB)$ for the linear functional $L$ as in \cref{sec:sos/def} for all $A,B \in \mathcal{O}_k$. Taking
\begin{align}\label{eq:def_S}
\mathcal{S} \defeq \ofc{I} \cup \ofc{\Pi_{ij}} \cup \ofc{\Pi_{ik} \mid (i,k) \in D},
\end{align}
our interest is in the matrix $M \in \mathbb{R}^{\mathcal{S} \times \mathcal{S}}$ with $M(A,B) \defeq \Gamma(A,B)$ for $A,B \in \mathcal{S}$. We have $M \succeq 0$ since $M=R\Gamma R^T$ for some matrix $R$. As a moment matrix over $\mathcal{S}$, $M$ is included in the first level of the swap or singlet projector hierarchies from \cite{takahashi2023,watts2024}.

Since we were able to assume $g_{ij} \in (-1,1)$, $g_{ik} > -1$ for $(i,k) \in D$, and $q_{ik} = g_{ik}$ for $(i,k) \in D \cup (i,j)$, we have
\begin{align}
\label{eq:M_pos_val}
M(\Pi_{ik},I) &= \frac{1 + q_{ik}}{2} > 0\text{ for all }(i,k) \in D,\\
\notag
M(\Pi_{ij},I) &= \frac{1+ q_{ij}}{2} \in (0,1).
\end{align}
By the above and the definition of $M$, 
\begin{align}
\label{eqn_firstLappearance}
M(A,A) = L(A^2) = L(A) = M(A,I) > 0, \text{ for all }A \in \mathcal{S},
\end{align}
where $L(A^2)=L(A)$ since $\mathcal{S}$ consists of projectors in $\mathcal{O}_2$. Consider a rescaling of $M$:
\begin{align*}
\td{M}(A,B) \defeq \frac{M(A,B)}{\sqrt{M(A,I)M(B,I)}}, \text{ for all }A,B \in \mathcal{S}.
\end{align*}
We will need three properties of $\td{M}$:
\begin{itemize}
    \item[(i)] $\td{M}(A,A)=1$ for all $A \in \mathcal{S}$, since $M(A,A)=M(A,I)>0$,
    \item[(ii)] $\td{M} \succeq 0$: $\td{M} = DMD$, where $D$ is the diagonal matrix with $D(A,A) = 1/\sqrt{M(A,I)}$, and
    \item[(iii)] $|\td{M}(A,B)| \leq 1/2$ for all distinct $A,B \in \mathcal{S}\setminus \ofc{I}$: This is established in \cite[Lemma~4.3]{takahashi2023}.
\end{itemize}

Consider the blocks of $\td{M}$ as induced by the sets of rows and columns corresponding to each of the three sets in \cref{eq:def_S}:
\begin{equation*}
\td{M}=
\begin{bmatrix}
K & \begin{matrix} v^T \\ u^T \end{matrix} \\
\begin{matrix} v & u \end{matrix} & N
\end{bmatrix},
\end{equation*}
where $N$ is a $(d_i-1)$ by $(d_i-1)$ matrix, $u$ and $v$ are vectors, and 
\begin{gather*}
K \defeq
\begin{bmatrix}
1 & p\\
p & 1
\end{bmatrix},\text{ with }\\
p \defeq \frac{M(\Pi_{ij},I)}{\sqrt{M(\Pi_{ij},I)M(I,I)}} = \sqrt{M(\Pi_{ij},I)} \in (0,1).
\end{gather*}
Analogously to the definition of $p$, observe that, in conjunction with \cref{eq:M_pos_val},
\begin{equation}\label{eq:def_vv}
v^Tv = \sum_{(i,k) \in D} M(\Pi_{ik},I) = \sum_{(i,k) \in D} \frac{1+q_{ik}}{2} > 0.
\end{equation}
We will derive the desired inequality from the positivity of $\td{M}$ through the Schur complement with respect to $K$:
\begin{align}
\notag
S &\defeq N - \begin{bmatrix} v & u \end{bmatrix} K^{-1} \begin{bmatrix} v^T\\ u^T \end{bmatrix}\\
  \notag
  &= N - \frac{1}{1-p^2} \begin{bmatrix} v & u \end{bmatrix} \begin{bmatrix} 1 & -p\\-p & 1 \end{bmatrix} \begin{bmatrix} v^T\\ u^T \end{bmatrix}\\
  \label{eq:schur}
  &= N - \frac{1}{1-p^2}\of{vv^T - p\of{vu^T+uv^T} + uu^T}. 
\end{align}
The positivity of $\td{M}$ is equivalent to that of $S$ so that \cref{eq:schur} implies
\begin{gather}
0 \leq v^TSv = v^TNv - \frac{1}{1-p^2}\of{\alpha^4 -2p \alpha^3 \beta \cos \theta + \alpha^2 \beta^2 \cos^2 \theta},\,\text{with}\\
\notag
\alpha \defeq \lVert v \rVert_2,\, \beta \defeq \lVert u \rVert_2,\,\text{ and }u^Tv=\alpha \beta \cos\theta.
\end{gather}
Dividing by $\alpha^2 = v^Tv > 0$ (\cref{eq:def_vv}) yields
\begin{align}\label{eq:alpha_cons2}
    \frac{1}{1-p^2}\of{\alpha^2 -2p \alpha \beta \cos\theta + \beta^2 \cos^2\theta} &\leq \frac{v^TNv}{v^Tv} 
    \leq \lambdaMax{N}
    \leq \frac{d_i}{2},
\end{align}
where the last inequality follows by Gershgorin's circle theorem and properties (i) and (iii) above.

Property (iii) applied to the entries of $u$ gives
\begin{align} \label{eq:alpha_cons1}
    \beta^2 = u^Tu \leq \frac{d_i-1}{4}.
\end{align}

The LHS of \cref{lem:degree_moe} with respect to the values $q$ is $2\alpha^2 - (d_i-1)$ by \cref{eq:def_vv}, and we will obtain the desired result by bounding $\alpha$ subject to \cref{eq:alpha_cons1,eq:alpha_cons2}. Letting 
\begin{equation}\label{eq:gamma_def}
\gamma \defeq \beta \cos\theta \in \of{-\sqrt{\frac{d_i-1}{4}},\,\sqrt{\frac{d_i-1}{4}}},
\end{equation}
we see from \cref{eq:alpha_cons2}:
\begin{gather*}
\alpha^2 -(2p \gamma)\alpha + \of{\gamma^2 - (1-p^2)\frac{d_i}{2}} \leq 0,\,\text{so}\\
\begin{aligned}
p\gamma - \sqrt{(1-p^2)\of{\frac{d_i}{2}-\gamma^2}} \leq \alpha &\leq p\gamma + \sqrt{(1-p^2)\of{\frac{d_i}{2}-\gamma^2}}\\
  &= \begin{pmatrix}p & \sqrt{1-p^2}\end{pmatrix} \begin{pmatrix}\gamma & \sqrt{\frac{d_i}{2}-\gamma^2}\end{pmatrix}^T\\
  &= \sqrt{\frac{d_i}{2}} \begin{pmatrix}p & \sqrt{1-p^2}\end{pmatrix} \begin{pmatrix}\td{\gamma} & \sqrt{1-\td{\gamma}^2}\end{pmatrix}^T\\
  &= \sqrt{\frac{d_i}{2}} \cos \phi,
\end{aligned}
\end{gather*}
for some $\phi$, where $\td{\gamma} = \gamma \sqrt{2/d_i}$. From the above, we always have the trivial bound $\alpha \leq \sqrt{d_i/2}$. This gives
\begin{equation*}
2\alpha^2 - (d_1-1) \leq 1,
\end{equation*}
which establishes \cref{lem:degree_moe} when $q_{ij} < -1/d_i$. If $q_{ij} \geq -1/d_i$,
\begin{equation*}
p = \sqrt{\frac{1+q_{ij}}{2}} \geq \sqrt{\frac{d_i-1}{2d_i}} = \sqrt{\frac{d_i-1}{4}} \sqrt{\frac{2}{d_i}} \geq \gamma \sqrt{\frac{2}{d_i}} = \td{\gamma}. 
\end{equation*}
In this case $\cos\phi$ is maximized when $\gamma = \sqrt{(d_i-1)/4}$, giving the bound,
\begin{equation*}
\alpha \leq \sqrt{p^2\frac{d_i-1}{4}}+\sqrt{(1-p^2)\frac{d_i+1}{4}}. 
\end{equation*}
Finally, this establishes \cref{lem:degree_moe}:
\begin{align*}
2\alpha^2-(d_i-1) &\leq p^2\frac{d_i-1}{2} + \sqrt{p^2(1-p^2)(d_i^2-1)} + (1-p^2)\frac{d_i+1}{2} - (d_i-1)\\
  &= \frac{1}{2}\of{3-d_i-2p^2 + \sqrt{4p^2(1-p^2)(d_i^2-1)}}\\
  &= \frac{1}{2}\of{2-d_i-q_{ij} + \sqrt{(d_i^2-1)\of{1-q_{ij}^2}}}.\qedhere
\end{align*}
\end{proof}

\section{Proof of main result}
\label{apx:proof_main_result}
\subsection{Verifying properties of \texorpdfstring{$\Theta$ and $\Lambda$}{Theta and Lambda}}

To prove our approximation ratio, we first must verify that $\Theta \in \mathcal{C}$ and $\Lambda \in \mathcal{D}_{\Theta, \beta}$. We do this in the following two claims:

\begin{claim}\label{claim:theta_in_c}
    The function $\Theta$ defined by points \cref{eq:theta_points} is in the set $\mathcal{C}$ defined in \cref{def:theta_set}.
\end{claim}
\begin{proof}
It is easy to see that $\Theta(0) = 0$ and $\Theta$ is monotonically increasing. Note that $\Theta$ is convex by inspection: it is piecewise linear, with increasing slope
$$
\frac{\gamma - 0}{Q(\beta) - 0} 
\le 0.2 
\le \frac{\frac{(\gamma/2 + \alpha'(1+\beta)-1)^2 }{1-\gamma} - \gamma}{\beta - Q(\beta)}
\le 0.393 
\le 
\frac{2(1-\alpha') - \frac{(\gamma/2 + \alpha'(1+\beta)-1)^2 }{1-\gamma}}{1 - \beta}\,.
$$    
It remains to show that $\Theta$ satisfies \cref{eqn:theta_condition}; that is, for all $x_1,x_2,\ldots,x_p \in \ofb{0,1}$ satisfying $\sum_{i=1}^p x_i \le 1$, we have
\begin{align}
\label{eqnASDfoihjw}
    \prod_{i=1}^p \of{1-\Theta(x_i)} \ge 1- \Theta\of{\sum_{i=1}^p x_i}.
\end{align}
\begin{itemize}
    \item Suppose there exists a $j \in [p]$ such that $1 - \Theta(x_j) = 0$. Then the left-hand side of \cref{eqnASDfoihjw} is zero. The right-hand side of \cref{eqnASDfoihjw} is $1-\Theta(\sum_{i \in [p]} x_i) \leq 1-\Theta(x_j) = 0$, where we have used that $\Theta$ is an increasing function. So $\Theta$ satisfies \cref{eqn:theta_condition} in this case.
    \item Otherwise, $1-\Theta(x_i) > 0$ for all $i \in [p]$. Since $\Theta$ is convex and $\Theta(0) =0$, the function $z(x)  \defeq \log(1-\Theta(x))$ is concave and satisfies $z(0) = 0$. It therefore holds for all $i \in [p]$ that
    $$
    z(x_i) = z\mleft( \frac{x_i}{\sum_{j \in [p]} {x_j}} \cdot \sum_{j \in [p]} {x_j} +0 \mright) \geq \frac{x_i}{\sum_{j \in [p]} {x_j}} \cdot z(\sum_{j \in [p]} {x_j}) + z(0) = \frac{x_i}{\sum_{j \in [p]} {x_j}} z(\sum_{j \in [p]} {x_j})\,.
    $$
    Hence,
    \begin{align*}
    \log \prod_{i\in[p]} \left( 1\!-\! \Theta(x_i) \right)
    =
    \sum_{i \in [p]} z(x_i) 
    \geq \sum_{i \in [p]} \frac{x_i}{\sum_{j \in [p]} {x_j}} \cdot z\big(\sum_{j \in [p]}\! {x_j} \big) 
    = z\big(\!\sum_{j \in [p]}\! {x_j} \big)
    =
    \log( 1\!-\! \Theta\big(\!\sum_{i \in [p]}\! x_i \big) )\,.
    \end{align*}
    Taking the exponential of both sides, we conclude $\Theta$ satisfies \cref{eqn:theta_condition} in this case. \qedhere
\end{itemize}

\end{proof}
\begin{claim}\label{claim:lambda_in_d}
    The function
    \begin{align*}
        \Lambda(y) = \frac{\of{\frac{1}{2}\Theta\of{Q(y)}+\alpha'(1+y)-1}^2}{1-\Theta\of{Q(y)}},
        \end{align*}
    where $\Theta$ is as in \cref{eq:theta_points},
        is in the set $\mathcal{D}_{\Theta,\beta}$ defined in \cref{def:lambda_set}. Moreover,
    \begin{align*}
    f(x,y) = \sqrt{\big(1-\Lambda(y^+)\big) \big(1-\Theta\left( Q(y^+)-x^+
    \right)\big)}\,,
    \end{align*}
    is decreasing as a function of $y$ for all $y\in(\beta, R(x)]$ and $x \in[-1,R(\beta)]$, where $R$ is as in \cref{eq:r_def}.
\end{claim}
\begin{proof}
    Since $f$ is continuous, it decreasing in $y$ implies $\Lambda \in \mathcal{D}_{\Theta, \beta}$. Let us define 
    \begin{align*}
        N(y) \defeq \of{\frac{1}{2}\Theta\of{Q(y)}+\alpha'(1+y)-1}^2,
    \end{align*}
    so that $\Lambda(y) = \frac{N(y)}{  1- \Theta(Q(y)) }$.   We show that $ 1 - \Theta(Q(y))- N(y) $ is decreasing in the interval $y\in(\beta,1]$:
    \begin{itemize}
        \item First consider $y \in [\sqrt{3}/2, 1]$. Then $Q(y) = 0$, and so $N(y) = (\alpha'(1+y)-1)^2$, and so $1 - \Theta(Q(y)) - N(y) = 1 - (\alpha'(1+y)-1)^2$. Then the derivative of $1 - \Theta(Q(y)) - N(y)$ is $-N'(y)=-2\alpha'(\alpha'(1+y)-1)$, which is negative for $y > 0.19 >  1/\alpha' - 1$.
    \item Now consider $y \in (\beta, \sqrt{3}/2]$. Here, $Q(y) = \frac{1}{2}\left(\sqrt{3\left(1-y^{2}\right)}-y\right)$, and $\Theta(Q(y)) = Q(y) \cdot \frac{\gamma}{Q(\beta)}$. Then the derivative of $1 - \Theta(Q(y)) - N(y)$ is
    \begin{align*}
    -\Theta'(Q(y)) Q'(y) - 2 \sqrt{N(y)} \left( \frac{1}{2} \Theta'(Q(y)) Q'(y) + \alpha' \right)
    \\
    = - \frac{\gamma}{Q(\beta)} \left(1 + \sqrt{N(y)} \right) Q'(y) - 2 \alpha \sqrt{N(y)}.
    \end{align*}
    The derivative of $Q$ is decreasing and negative in this region, since
    $$
    Q'(y) = \frac{1}{2} \left(\frac{-2y\sqrt{3}}{2\sqrt{1-y^2}} -1\right) = -\frac{1}{2} \left(1 +  \frac{y\sqrt{3}}{\sqrt{1-y^2}} \right).
    $$
    So the derivative of $1 - \Theta(Q(y)) - N(y)$ is at most 
    $$
    -\frac{\gamma Q'(\sqrt{3}/2)}{Q(\beta)} \left(1 + \sqrt{N(y)} \right) - 2 \alpha \sqrt{N(y)} = \frac{2 \gamma}{Q(\beta)} + \sqrt{N(y)} \left( \frac{2 \gamma}{Q(\beta)} - 2 \alpha\right)\,.
    $$
    Since $N$ is increasing in this region, the expression above is monotonic in $y$. At $y = \beta$ it is $< -0.25$ and at $y = \sqrt{3}/2$ it is $<-0.45$. So the derivative  of $1 - \Theta(Q(y)) - N(y)$ is negative in this region.
    \end{itemize}
    We prove this claim in two steps. First, assume $x \le 0$.
    In this case, we have 
    \begin{align*}
        f(x,y)  = \sqrt{\left(1-\frac{N(y^+)}{1-\Theta(Q(y^+))} \right) \left( 1 -\Theta(Q(y^+))\right)}
        =\sqrt{1 - \Theta(Q(y^+)) - N(y^+)
        }\,,
    \end{align*}
    which is decreasing in $y$ for all $y \in (\beta, 1]$.

    From here, we assume $x > 0.$ We have that
    \begin{align*}
        f(x,y)^2 
        &= (1- \Lambda (y^+))(1 - \Theta(Q(y^+) - x^+) \\
        &= ({1- \Theta(Q(y^+ )) -N(y^+) })\cdot 
        \frac{1 - \Theta(Q(y^+) - x^+)}{1 - \Theta(Q(y^+))}.
    \end{align*}
    The first factor is decreasing for $y\in(\beta,1]$ and positive, since $1 - \Theta(Q(1)) - N(1) = 1 - (2\alpha' - 1)^2$ is $>0.5$.
    To see that the second factor is decreasing,
    consider that $Q$ is decreasing for $y \in [0,1]$ and
    \begin{align}
    \label{eq:fn_in_D_temp}
        \frac{1- \Theta(y -a)}{1- \Theta(y)},
    \end{align}
    is increasing for $y \in [a,1]$ (at any constant $a \in[0,1]$) for the following reasons:

    The function in \cref{eq:fn_in_D_temp} is constant if $a = 0$. Thus, assume that $a \in (0,1]$. Define the concave function $z(t) \defeq \log{\left(1-\Theta(t) \right)}$. Then
    \begin{align}
    \label{eq:fn_in_D_temp_2}
        \log{\frac{1- \Theta(t -a)}{1- \Theta(t)} } = z(t-a) - z(t) = -a \frac{z(t) - z(t-a)}{a}\,.
    \end{align}
    The fraction $\frac{z(t) - z(t-a)}{a}$ in \cref{eq:fn_in_D_temp_2} equals the slope of the line segment connecting the points $(t-a,z(t-a))$ and $(t,z(t))$. Note that $z$ is concave (because $\Theta$ is convex; see also the proof of \Cref{claim:theta_in_c}). So, this fraction $\frac{z(t) - z(t-a)}{a}$ is decreasing in $t$. Thus, $\log{\frac{1- \Theta(t -a)}{1- \Theta(t)} }$ is increasing in $t$, and therefore $\frac{1- \Theta(t -a)}{1- \Theta(t)}$ is also increasing in $t$.
\end{proof}

\subsection{Completing the proof}
We now finish the proof of the main result.
\begin{proof}[Proof of \Cref{thm:apx_ratio_formal}]
Consider the choice of $\Theta$, $\Lambda$, $\beta$, and $\gamma$ in the statement of \Cref{thm:apx_ratio_formal}. By \Cref{claim:theta_in_c}, $\Theta \in \mathcal{C}$, and by \Cref{claim:lambda_in_d}, we have $\Lambda \in \mathcal{D}_{\Theta, \beta}$ with our choice of $\Theta$ and $\beta > \frac{1}{2}$. As a result, we can invoke \Cref{claim:all_cases} and study the three expressions in \cref{eq:minimization_problem}. Below we show that these functions are always at least $\alpha' \ge 0.839511$. Code for all stated computational evaluations is publicly available\footnote{\url{https://github.com/jamessud/EPR_0.839511_approximation_ratio/blob/main/epr_8395_approx.ipynb}.}.

Our general strategy is as follows. Recall that $\Theta$ is piecewise linear; we denote the three pieces as $\Theta_1, \Theta_2, \Theta_3$ respectively. Also, $Q$ is a piecewise function. The functions $r_1$, $r_2$ and $r_3$ are all compositions of $Q$ and $\Theta$. We analyze these functions by breaking them down into pieces, whose domains are determined by the breakpoints of functions of $Q$ and $\Theta$. For instance, we can compute
\begin{align}
    \Theta\of{Q(g^+)} &=
    \begin{cases}
        \Theta_3\of{1}, & g \in \ofb{-1,0} \\
        \Theta_3\of{1-g}, & g\in \ofb{0,1-\beta} \\
        \Theta_2\of{1-g}, & g\in \ofb{1-\beta, \frac{1}{2}} \\
        \Theta_2\of{R\of{g}}, & g\in\ofb{\frac{1}{2},\beta}\\
        \Theta_1\of{R\of{g}}, & g\in\ofb{\beta, \frac{\sqrt{3}}{2}}\\
        0, & g\in\ofb{\frac{\sqrt{3}}{2}, 1}.\\
    \end{cases}\label{eq:t_q_pieces}
\end{align}
\paragraph{Case $r_1$:}
We use \cref{eq:t_q_pieces} and the pieces of $\Theta$ to write $r_1$ as
\begin{align*}
    r_1(g) &\defeq \begin{cases}
        r_{1a}(g), & g\in\ofb{0,Q(\beta)},\\
        r_{1b}(g), & g\in\ofb{Q(\beta),1-\beta},\\
        r_{1c}(g), & g\in\ofb{1-\beta,\frac{1}{2}},\\
        r_{1d}(g), & g\in\ofb{\frac{1}{2},\beta},\\
    \end{cases}\\
    r_{1a}(g) &\defeq \frac{2-\Theta_3(1-g)+2\sqrt{\Theta_1(g)\,(1-\Theta_3(1-g))}}{2(1+g)},\\
    r_{1b}(g) &\defeq \frac{2-\Theta_3(1-g)+2\sqrt{\Theta_2(g)\,(1-\Theta_3(1-g))}}{2(1+g)}  \\
    r_{1c}(g) &\defeq \frac{2-\Theta_2(1-g)+2\sqrt{\Theta_2(g)\,(1-\Theta_2(1-g))}}{2(1+g)}  \\
    r_{1d}(g) &\defeq \frac{2-\Theta_2(R(g))+2\sqrt{\Theta_2(g)\,(1-\Theta_2(R(g)))}}{2(1+g)}.
\end{align*}
We now show that each piece is at least $\alpha'$. We begin with $r_{1a}$ through $r_{1c}$. We will show that each of these has one critical point in $g \ge 0$. We analyze $r_{1a}$; the analysis for $r_{1b}$ and $r_{1c}$ are nearly identical. To find the critical points consider the function
\begin{align*}
    r(g) = \frac{N(g)}{2(1+g)},
\end{align*}
then the derivative is
\begin{align*}
    r'(g) = \frac{(1+g)N'(g)-N(g)}{2(1+g)^2}.
\end{align*}
So critical points occur when 
\begin{align}
    (1+g)N'(g)=N(g). \label{eq:critical_points}
\end{align}
Let $h_a(g) \defeq \Theta_1(g)(1 - \Theta_3(1-g))$. Then when $h_a(g) > 0$, the critical points are found by plugging the numerator of $r_{1a}$ into \cref{eq:critical_points}
$$
(1+g) \left( \Theta_3'(1-g) + \frac{h_a'(g)}{\sqrt{h_a(g)}} \right) = 2 - \Theta_3(1-g) + 2 \sqrt{h_a(g)}\,.
$$
We can multiply both sides by $\sqrt{h_a(g)}$ and rearrange terms:
$$
 (1+g) h_a'(g) -2 h_a(g) = \left(2 - \Theta_3(1-g) - (1+g) \Theta_3'(1-g)\right) \sqrt{h_a(g)}.
$$
Since $h_a$ is a quadratic, the expression $(1+g) h_a'(g)- 2 h_a(g)$ is linear in $g$. Similarly, since $\Theta_3$ is linear, the expression $\left(2 - \Theta_3(1-g) - (1+g) \Theta_3'(1-g)\right)$ is constant. By squaring both sides, we see this equation is zero when two different quadratics are equal, which can occur at most twice.

It is easy to verify that $h_a$ is a positive quadratic with a root in $g < 0$ and a root at $g = 0$. So $h_a(g) > 0$ for $g > 0$. By inspection, we observe that $r_{1a}'(g)= 0$ at $g \approx -1.76$ and $g \approx 0.06$. By the above, there are no other critical points. By inspection, $r_{1a}'(0.05) > 0$. So the minimum value of $r_{1a}$ in $0 \le g \le Q(\beta)$ occurs at an endpoint, which take values $\alpha'$ and $>0.839529$, respectively.  

We may repeat this analysis for $h_b \defeq \Theta_2(g) (1 - \Theta_3(1-g))$ and $h_c  \defeq \Theta_2(g) (1 - \Theta_2(1-g))$, which implies that $r_{1b}$ and $r_{1c}$ each have at most two critical points.
\begin{itemize}
\item It is easy to verify that $h_b$ is a positive quadratic with a root in $g < 0$ and a root in $0 < g < 0.19 < Q(\beta)$.  So $h_b(g) > 0$ for $g \ge Q(\beta)$. By inspection, we observe that $r'_{1b}(g) = 0$ at $g \approx -1.7$ and $g \approx 0.427$. By the above, there are no other critical points. So the minimum value of $r_{1b}$ in $Q(\beta) \le g \le 1 - \beta$ occurs at an endpoint, which take values $> 0.839529$ and $> 0.842$, respectively.
\item It is easy to verify that $h_c$ is a positive quadratic with a root in $g < 0$ and a root in $0 < g < 0.19 < 1-\beta$.  So $h_c(g) > 0$ for $g \ge 1-\beta$. By inspection, we observe that $r'_{1c}(g) = 0$ at $g \approx -1.7$ and $g \approx 0.427$, and $r_{1c}'(0.4) > 0$. By the above, there are no other critical points. So the minimum value of $r_{1c}$ in $1-\beta \le g \le \frac{1}{2}$ occurs at an endpoint, which take values $> 0.842$ and $> 0.845$, respectively.
\end{itemize}

We now consider the last piece $r_{1d}$. Let $h(g) \defeq \Theta_2(g)(1 - \Theta_2(R(g)))$.  The derivative of $r_{1d}$ is
$$
r_{1d}'(g) = \frac{(1+g) \left( -\Theta_2'(R(g))R'(g) + h'(g)/\sqrt{h(g)} \right) - 2 + \Theta_2(R(g)) - 2 \sqrt{h(g)}}{2(1+g)^2}\,.
$$
We will show $r_{1d}'(g) < 0$ in the region $\frac{1}{2} \le g \le \beta$. The numerator has the form $(1+g)f'(g) - f(g)$. The derivative of this has the form $(1+g)f''(g)$, i.e.
\begin{align}
\label{eqn:bound_deriv}
(1+g) \left( -\Theta_2'(R(g))R''(g) + \frac{h''(g)}{\sqrt{h(g)}} - \frac{h'(g)^2}{2h(g)^{3/2}} \right)\,.
\end{align}
We use crude bounds to show this value is not too large in the region $\frac{1}{2} \le g \le \beta$, and then apply Taylor's remainder theorem. Starting with $R$ and its derivatives:
\begin{align*}
R'(g) &= - \frac{\sqrt{3}}{2} \frac{x}{\sqrt{1 -  x^{2}}} - \frac{1}{2}\\
R''(g) & = -\frac{\sqrt{3}}{2} (1-g^2)^{-3/2}
\end{align*}
So $R(g) \in [0.3, 0.5]$, $|R'(g)| \in [1, 2]$, $|R''(g)| \in [1, 2.5]$.
Next we bound $h$:
\begin{align*}
h'(g) &= \Theta_2' \cdot (1 - \Theta_2(R(g)) - \Theta_2(g) R'(g)) \\
h''(g) &= -\Theta_2' \cdot (2 \Theta_2' R'(g)  + \Theta_2(g) R''(g))
\end{align*}
(Recall that $\Theta_2'$ is a constant.) So $\Theta_2(g) \in [0.1, 0.2]$, $\Theta_2(R(g)) \in [0, 0.2]$, $|\Theta_2'| \le 0.4$, and $h(g) \in [0.08, 0.2]$. Then $|h'(g)| \le 0.6$, and $|h''(g)| \le 1.5$.
Then the value in \cref{eqn:bound_deriv} is at most 
$$
2 \cdot \left(1 + \frac{1.5}{\sqrt{0.08}} + \frac{0.6^2}{2 \cdot 0.08^{3/2}}\right) \le 30\,.
$$
We invoke Taylor's remainder theorem on $n(g) \defeq r'_{1d}(g) \cdot 2(1+g)^2$. For any points $\frac{1}{2} \le a, g\le \beta$, we have $|n(a) - n(g)| \le n'(a) \cdot |a-g| \le 30 \cdot |a-g|$. We computationally evaluate $n$ on 300 equally spaced points in the interval $[\frac{1}{2}, \beta]$; by Taylor's remainder theorem, $n$ is at most $30 \cdot \frac{\beta - 0.5}{300} < 0.1$ plus the largest evaluation. All evaluations are $< -0.11$, so $n$ is negative in the region. So $r'_{1d}(g) < 0$ in the region.
The smallest value of $r_{1d}$ thus occurs at $g = \beta$, where 
$$
r_{1d}(\beta) = \frac{2 - \Theta(R(\beta)) + 2 \sqrt{\Theta(\beta)(1 - \Theta(R(\beta))}}{2(1+\beta)} = \frac{2 - \gamma + 2 (\gamma/2 + \alpha'(1+\beta)-1)}{2(1+\beta)} = \alpha'\,.
$$

\paragraph{Case $r_2$:} This expression is always equal to $\alpha'$ when $g\ge \beta$ by definition of $\Lambda$: 
    \begin{align*}
        \frac{2 - \Theta(Q(g)) + 2\sqrt{\Lambda(g)(1 - \Theta(Q(g)))}}{2(1+g)} = \frac{2 - \Theta(Q(g)) + 2 \cdot \left( \frac{1}{2} \Theta(Q(g)) + \alpha'(1+g) - 1 \right)}{2(1+g)} = \alpha'\,.
    \end{align*}

\paragraph{Case $r_3$:}
For $r_3$, we must consider the pieces of expressions containing $f^*$, which depends on $\Lambda\of{R(\cdot)^+}$, which in turn depends on the pieces of $Q$ and $\Theta$. We build up the pieces of these expressions, starting with
\begin{align*}
    Q\of{R(g)^+} = 
    \begin{cases}
        -R(-g), &g\in\ofb{-1,-\frac{\sqrt{3}}{2}},\\
        0, &x\in\ofb{-\frac{\sqrt{3}}{2},0},\\
        g, &g\in\ofb{0,\frac{1}{2}},\\
        1-R(g), &g\in\ofb{\frac{1}{2},\frac{\sqrt{3}}{2}},\\
        1, &g\in\ofb{\frac{\sqrt{3}}{2},1}.\\
    \end{cases}
\end{align*}
Here, we used the simplifications $R(R(g))=g$ when $g\ge 0$ and $R(R(g)) = -R(-g)$ when $g < -1/2$. We then compute the pieces
\begin{align}
    \Theta\of{Q\of{R(g)^+}} &= 
    \begin{cases}
        \Theta_2\of{-R(-g)}, &g\in\ofb{-1,-\delta_1},\\
        \Theta_1\of{-R(-g)}, &g\in\ofb{-\delta_1,-\frac{\sqrt{3}}{2}},\\
        0, & g\in\ofb{-\frac{\sqrt{3}}{2}, 0}, \\
        \Theta_1\of{g}, &g\in\ofb{\,0,Q(\beta)},\\
        \Theta_2\of{g}, &g\in\ofb{Q(\beta), \frac{1}{2}},\\
        \Theta_2\of{1-R(g)}, &g\in\ofb{\frac{1}{2},\delta_2},\\
        \Theta_3\of{1-R(g)}, &x\in\ofb{\delta_2,\frac{\sqrt{3}}{2}},\\
        \Theta_3\of{1}, &g\in\ofb{\frac{\sqrt{3}}{2},1},\\
    \end{cases} \label{eq:t_q_r_pieces} \\
    \delta_1 &\defeq \beta + R(\beta) \approx 0.9779, \nonumber\\
    \delta_2 &\defeq R(1-\beta) \approx 0.6525 \nonumber.
\end{align}
We likewise expand the expressions 
\begin{align*}
        Q\of{R(g)^+}-g^+ &= 
    \begin{cases}
        -R(-g), &g\in\ofb{-1,-\frac{\sqrt{3}}{2}},\\
        0, &g\in\ofb{-\frac{\sqrt{3}}{2},\frac{1}{2}},\\
        1-R(g)-g, &g\in\ofb{\frac{1}{2},\frac{\sqrt{3}}{2}},\\
        1-g, &g\in\ofb{\frac{\sqrt{3}}{2},1}.\\
    \end{cases} \\
    \Theta\of{Q\!\of{R(g)^+}-g^+} &= 
    \begin{cases}
        \Theta_2\of{-R(-g)}, &g\in\ofb{-1,-\delta_1},\\
        \Theta_1\of{-R(-g))}, &g\in\ofb{-\delta_1,-\frac{\sqrt{3}}{2}},\\
        0, & g\in\ofb{-\frac{\sqrt{3}}{2}, \frac{1}{2}}, \\
        \Theta_1\of{1-R(g)-g}, &g\in\ofb{\frac{1}{2},\frac{\sqrt{3}}{2}},\\
        \Theta_1\of{1-g}, &g\in\ofb{\frac{\sqrt{3}}{2},1}.
    \end{cases}
\end{align*}
Now consider $\Lambda$. The pieces of $\Lambda$ are fully determined by the pieces of $\Theta\of{Q\of{g}}$. In Case $r_3$, we only ever use $\Lambda\of{R\of{\cdot}^+}$, so it suffices to consider the pieces of $\Theta\of{Q\of{R\of{\cdot}^+}}$. These pieces are exactly given by \cref{eq:t_q_r_pieces}. The final expression that appears is $f^*$. This function depends on both $\Lambda\of{R\of{\cdot}^+}$ and $\Theta\of{Q\!\of{R(g)^+}-g^+}$. We just argued that the pieces of $\Lambda\of{R\of{g}^+}$ are given by \cref{eq:t_q_r_pieces} and we note that the breakpoints defining $\Theta\of{Q\!\of{R(g)^+}-g^+}$ are all contained in the breakpoints of  \cref{eq:t_q_r_pieces}. Thus, the pieces of \cref{eq:t_q_r_pieces} specify all the pieces of $f^*$ as well. 

The remaining function that arises in case $r_3$ is $\Theta\of{g^+}$. The breakpoints of this are by definition $0$, $Q(\beta)=R(\beta)$, $\beta$, and $1$. This introduces the new breakpoint $\beta$. However, note that we only need to consider case $r_3$ in the domain $-1 < g \le R(\beta)$. We can see by inspection that $\beta$ falls outside this domain. In fact, this domain corresponds to the first four pieces of \cref{eq:t_q_r_pieces}, We refer to these regions as $a\defeq (-1,-\delta_1]$, $b\defeq[-\delta_1, -\sqrt{3}/2]$, $c\defeq[-\sqrt{3}/2, 0]$, $d\defeq [0, Q(\beta)]$.

It will be helpful to name the pieces of $f^*$ in the regions $a, b, c, d$. We label the functions by the subscripts corresponding to regions:
\begin{align*}
        f^*(g) &= \begin{cases}
        f^*_a(g), &g\in\ofb{-1,-\delta_1}\,,\\
        f^*_b(g), &g\in\ofb{-\delta_1,-\frac{\sqrt{3}}{2}}\,,\\
        f^*_c(g), & g\in\ofb{-\frac{\sqrt{3}}{2}, 0}\,, \\
        f^*_d(g), &g\in\ofb{\,0,R(\beta)}\,,
    \end{cases} \\
    f^*_a(g) &\defeq \sqrt{\of{1 - \Theta_2\of{-R(-g)}} - \left(\frac{1}{2} \cdot \Theta_2(-R(-g)) + \alpha'(1+R(g))-1\right)^2}\,,\\
    f^*_b(g) &\defeq \sqrt{\of{1 - \Theta_1\of{-R(-g)}} - \left(\frac{1}{2} \cdot \Theta_1(-R(-g)) + \alpha'(1+R(g))-1\right)^2}\,,\\
    f^*_c(g) &\defeq \sqrt{1 - \left( \alpha'(1+R(g))-1\right)^2}\,,\\
    f^*_d(g) &\defeq \sqrt{1 - \frac{\left(\frac{1}{2} \cdot \Theta_1(g) + \alpha'(1+R(g))-1\right)^2}{1 - \Theta_1(g)}}\,.
\end{align*}
With this notation we now perform the analysis for each case over these four pieces for $r_3$ 
\begin{align*}
    r_{3a}(g) &\defeq \frac{1 + f_a^*(g)^2}{2(1+g)}, \\
    r_{3b}(g) &\defeq \frac{1 + f_b^*(g)^2}{2(1+g)}, \\
    r_{3c}(g) &\defeq \frac{1 + f_c^*(g)^2}{2(1+g)}, \\
    r_{3d}(g) &\defeq \frac{1 + f_d^*(g)^2 + 2 \sqrt{\Theta_1(g)} f_d^*(g)}{2(1+g)}.
\end{align*}
We start with $r_{3a}$ and $r_{3b}$. Both $f_a^*$ and $f_b^*$ have the form $\sqrt{1 - x - (\frac{x}{2} + m)^2}$. This is the square root of a negative quadratic in $x$ with midpoint $-2(m+1)$. In the region $g \in (-1, -\frac{\sqrt{3}}{2}]$, $R(g) \ge \frac{1}{2}$, and so $m+1 > 0$; so both $f_a^*$ and $f_b^*$ are decreasing in $\Theta(-R(-g))$. So we may lower bound $f_a^*$ and $f_b^*$ with an upper bound of $\Theta(-R(-g)) < 0.2$:
\begin{align}
    f_a^*(g) , f_b^*(g) \ge \sqrt{0.8 - (1.5\alpha'-0.6)^2} > 0.6\,\label{eq:f_star_lower_bound}.
\end{align}
So $r_{3a}(g), r_{3b}(g) \ge \frac{1 + 0.6^2}{2(1-\sqrt{3}/2)} > 1 > \alpha'$.

We now handle $r_{3c}$. This simplifies to
$$
r_{3c}(g) = \frac{2 - (\alpha'(1+R(g)) - 1)^2}{2(1+g)}\,.
$$
By inspection, $r_{3c}(0) > \alpha'$. We now show the expression is decreasing in the region $g \in [-\frac{\sqrt{3}}{2}, 0]$, and so greater than $\alpha'$ in the region. The derivative is
$$
r_{3c}'(g) = \frac{1}{2(1+g)^2} \cdot \left(-2\alpha' R'(g)(1+g)(\alpha'(1+R(g)) - 1) - 2 + (\alpha'(1+R(g)) - 1)^2 \right),
$$
Since $0 \le R(g) \le 1$ in this region, $(\alpha'(1 + R(g))-1)^2 \le (2\alpha'-1)^2 < 0.5$. So the derivative is negative if 
$$
|R'(g)(1+g)| < 1.5/\alpha'\,.
$$
Recall that 
\begin{align*}
    R'(g) = -\frac{1}{2} - \frac{\sqrt{3}g}{2\sqrt{1-g^2}},
\end{align*}
which is decreasing. Then  $|R'(g)(1+g)| \le \max\{|R'(0)|, |R'(-\frac{\sqrt{3}}{2})|\}$, which equals $1  < 1.5/\alpha'$. So $r_{3c}$ is decreasing, and so takes value at least $r_{3c}(0) > \alpha'$.

We finally consider $r_{3d}$. Our strategy is to show that the function is increasing in the domain $\ofb{0,\delta}$, for some small $\delta$. Thus, it is minimized in this domain at $g=0$. For the rest of the domain $d$ we bound the derivative and apply Taylor's remainder's theorem. We first offer crude bounds on $f_d^*$ and its derivative. In the region $g \in [0, Q(\beta)]$, $R(g) \in [\beta,  \frac{\sqrt{3}}{2}]$ and $\Theta_1(g) \in [0, 0.05]$. So $\left(\frac{1}{2}\cdot\Theta_{1}(g)+\alpha'(1+R(g))-1\right)^2  \in [(\alpha'(1+\beta)-1)^2, (0.025+\alpha'(1+\frac{\sqrt{3}}{2})-1)^2] \subseteq[0.16, 0.35]$. So $f_d^* \in [\sqrt{1 - 0.16}, \sqrt{1 - \frac{0.35}{0.95}}] \subseteq [0.79, 0.92]$.

The derivative $f_d^{*\,\prime}$ takes value
\begin{align*}
    f_d^{*\,\prime} &= \frac{1}{2 f_d^*} \cdot \frac{ -1}{(1-\Theta_1(g))^2} \cdot &\Big[ 2(1-\Theta_1(g))\left(\frac{1}{2}\cdot\Theta_{1}(g)+\alpha'(1+R(g))-1\right)\of{\frac{1}{2}\cdot\Theta_{1}' + \alpha' R'(g)} \\
    & &+ \Theta_1' \cdot \left(\frac{1}{2}\cdot\Theta_{1}(g) +\alpha'(1+R(g))-1\right)^2 \Big].
\end{align*}
The expression in the bracket is negative, since $2 \cdot 0.95 \cdot \sqrt{0.16} \cdot (0.025 + \alpha' \cdot (-0.5)) < -0.3$, and $\Theta_1' \in [0.15, 0.16]$, so $0.16 \cdot 0.35 = 0.056$. So $f_d^{*\,\prime} >  \frac{1}{2 \cdot 0.92} \cdot \frac{1}{0.95^2} \cdot (0.3 - 0.056) > 0.14$. A similar calculation shows $f_d^{*\,\prime} < \frac{-1}{2 \cdot 0.79} \cdot ( 2 \cdot \sqrt{0.35} \cdot (-\alpha') + 0.15 \cdot 0.16) < 0.62$. So $f_d^{*\,\prime} \in [0.14, 0.62]$.

The derivative of $r_{3d}$ is $0.5(1+g)^{-2}$ times the quantity
\begin{align*}
  (1+g)\cdot\of{2 f_d^{*\,\prime}(g)\Big(f_d^*(g) + \sqrt{\Theta_1(g)}\Big) + \frac{\Theta_1'(g)}{\sqrt{\Theta_1(g)}}\, f_d^*(g)} 
    - \left( 1 + f_d^*(g)^2 + 2\sqrt{\Theta_1(g)} f_d^*(g) \right)\,.
\end{align*}
Using the bounds on the ranges of $f^*$ and $\Theta_1$, as well as their derivatives, this quantity is at least
$$
2 \cdot 0.14 \cdot 0.79 + 0.15 \cdot \frac{0.79}{ \sqrt{\Theta_1(g)}} - (1 + 0.92^2 + 2 \cdot 0.92 \cdot \sqrt{0.05}) > \frac{0.1185}{\sqrt{\Theta_1(g)}} - 2.04\,.
$$
Recall that $\Theta_1(g) < 0.16g$. So when $g < 0.02$, the quantity is positive, and $r_{3d}$ is increasing. So $r_{3d}(g) > r_{3d}(0) > \alpha'$ for all $g \in [0, 0.02)$.

When $g \in [0.02, Q(\beta)]$, we invoke Taylor's remainder theorem. First, note that the derivative $r_{3d}'(g)$ takes value at most
$$
2 \cdot 0.62 \cdot 0.92 + 0.16 \cdot \frac{0.92}{ \sqrt{\Theta_1(g)}} + (1 + 0.92^2 + 2 \cdot 0.92 \cdot \sqrt{0.05}) < \frac{0.1472}{\sqrt{\Theta_1(g)}} + 3.4\,.
$$
When $g \in [0.02, Q(\beta)]$, $\Theta_1(g) > 0.003$, and so 
$|r_{3d}'(g)|< 6.1$.  So, for any points $0.02 \le a, g\le Q(\beta)$, we have $|r_{4d}(a) - r_{4d}(g)| \le r_{3d}'(a) \cdot |a-g| \le 6.1 \cdot |a-g|$. We computationally evaluate $r_{3d}$ on 1000 equally spaced points in the interval $[0.02, Q(\beta)]$; by Taylor's remainder theorem, $r_{3d}$ is at least the smallest evaluation minus $6.1 \cdot \frac{Q(\beta)-0.02}{1000} < 0.002$. All evaluations are $> 0.842$, so $r_{3d}(g) > 0.840 > \alpha'=0.839511$ in this region.
\end{proof}

\clearpage
\newpage

\section{Notation}
\label{apx:notation}
For ease of reference, we summarize our notation in the following \cref{tab:notation}.
\begin{table}[ht]
\centering
\begin{tabular}{|c|c|c|}
\hline
\textbf{Value}    & \textbf{Expression} & \textbf{Meaning} \\
\hline
$\alpha'$ &    $0.839511$        &   Target approximation ratio      \\ 
\hline
   $\beta$      &    $0.67$        &    Hyperparameter     \\
   \hline
    $\gamma$     &     $0.049$       &      Hyperparameter \\ 
    \hline 
    $R(x)$     &    $\vphantom{\bigg(}\frac{1}{2}\left( \sqrt{3\left(1-x^2\right)} - x \right)$       &  MoE bound (for $g$) from \Cref{claim:moe_lp24} \\ 
    \hline
    $Q(x)$     &   $\begin{cases} 1-x, \quad x \in [0,\frac{1}{2}]\\ R(x), \quad\, x \in (\frac{1}{2}, \frac{\sqrt{3}}{2}] \\ 0, \quad\quad\;\;\, x \in (\frac{\sqrt{3}}{2}, 1]\end{cases}$       &      \makecell{MoE bound (for $g^+$) \\ from \Cref{cor:moe}} \\  
    \hline
    $\Theta(x)$     &    \makecell{Piecewise linear function through \\ $\{(0,0), (Q(\beta), \gamma), (\beta, \frac{(\gamma/2 + \alpha'(1+\beta)-1)^2}{1-\gamma}), (1, 2-2\alpha') \}$ \\
    (we denote the lines as $\Theta_1(x)$, $\Theta_2(x)$, and $\Theta_3(x)$)
    }      
    &     \makecell{Computes $\sin^2 \theta$ from SDP value \\ (when $g\le \beta$)} \\ 
    \hline
    $\Lambda(x)$     &    $\frac{\left(\frac{1}{2} \cdot \Theta(Q(x)) + \alpha'(1+x)-1\right)^2}{1 - \Theta(Q(x))}$      &     \makecell{Computes $\sin^2 \theta$ from SDP value \\ (when $g >\beta$)} \\ 
    \hline
    $\nu(g)$ & $ \begin{cases}
        \arcsin \sqrt{\Theta(g^+)}, & \text{if}\; x \le \beta \\
       \arcsin \sqrt{\Lambda(g^+)}, & \text{if} \;x > \beta
    \end{cases}$ & Computes angle $\theta$ from SDP value $g$ \\  
    \hline
    $f^*(x)$ & $\vphantom{\Big\{}\sqrt{(1 - \Lambda(R(x)^+))(1 - \Theta(Q(R(x)^+) - x^+))}$ & Function used in analysis \\
    \hline
    $r_1(g)$ & 
    $\frac{2-\Theta(Q(g^+))+2\sqrt{\Theta(g^+)\,(1-\Theta(Q(g^+)))}}{2(1+g)}$
    & \makecell{Expression in \Cref{claim:all_cases}  \\ (minimized over $g \in [0, \beta]$)} \\
    \hline
    $r_2(g)$  & 
    $\frac{2-\Theta(Q(g^+))+2\sqrt{\Lambda(g^+)\,(1-\Theta(Q(g^+)))}}{2(1+g)}$
    &  \makecell{Expression in \Cref{claim:all_cases}  \\ (minimized over $g \in [\beta, 1]$)}\\
    \hline
    $r_3(g)$ &
    $\frac{\,1 + {f^*(g)}^2 \;+\; 2\sqrt{\Theta(g^+)}f^*(g)}{2(1+g)}$
    &  \makecell{Expression in \Cref{claim:all_cases} \\ (minimized over $g \in (-1, R(\beta)]$)} \\
    \hline
\end{tabular}
\caption{Table of important variables and functions, and their use in this paper.}
\label{tab:notation}
\end{table}

\end{document}